\documentclass[aps,showpacs,amsmath,amssymb,twocolumn,pra,superscriptaddress,notitlepage]{revtex4-2}

\usepackage{graphicx}
\usepackage{dcolumn}
\usepackage{subcaption}
\usepackage{booktabs}
\usepackage{bm}
\usepackage{braket}
\usepackage[dvipsnames]{xcolor}
\usepackage{amsfonts}
\usepackage{amsthm}
\usepackage{hyperref}
\usepackage{cleveref}
\crefname{equation}{Eq.}{Eqs.}
\Crefname{equation}{Eq.}{Eqs.}
\crefname{figure}{Fig.}{Figs.}
\Crefname{figure}{Fig.}{Figs.}
\crefname{subfigure}{Fig.}{Figs.}
\Crefname{subfigure}{Fig.}{Figs.}
\crefname{table}{Table}{Tables}
\Crefname{table}{Table}{Tables}
\crefname{section}{Sec.}{Secs.}
\Crefname{section}{Sec.}{Secs.}
\crefname{subsection}{Sec.}{Secs.}
\Crefname{subsection}{Sec.}{Secs.}
\crefname{subsubsection}{Sec.}{Secs.}
\Crefname{subsubsection}{Sec.}{Secs.}
\crefname{appendix}{Appendix}{Appendices}
\Crefname{appendix}{Appendix}{Appendices}
\crefname{theorem}{Theorem}{Theorems}
\Crefname{theorem}{Theorem}{Theorems}
\usepackage{amsmath,bm}
\usepackage{color}
\usepackage{algorithm}
\usepackage{physics}
\usepackage{mathtools}
\usepackage{bbm}

\newtheorem{theorem}{Theorem}

\definecolor{grey}{rgb}{.5,.5,.5}

\DeclareMathOperator{\E}{\mathbb{E}}

\begin{document}


\title{Error-Mitigated Hamiltonian Simulation: Complexity Analysis and Optimization for Near-Term and Early-Fault-Tolerant Quantum Computers}


\author{Keisuke Murota}
\email{keisuke.murota@phys.s.u-tokyo.ac.jp}
\affiliation{Department of Physics, The University of Tokyo, Tokyo 113-0033, Japan}

\author{Synge Todo}
\email{wistaria@phys.s.u-tokyo.ac.jp}
\affiliation{Department of Physics, The University of Tokyo, Tokyo 113-0033, Japan}
\affiliation{Institute for Physics of Intelligence, The University of Tokyo, Tokyo 113-0033, Japan}
\affiliation{Institute for Solid State Physics, The University of Tokyo, Kashiwa, 277-8581, Japan}

\author{Suguru Endo}
\email{suguru.endou@ntt.com}
\affiliation{NTT Computer and Data Science Laboratories, NTT Inc., Musashino 180-8585, Japan}
\affiliation{NTT Research Center for Theoretical Quantum Information, NTT Inc. 3-1 Morinosato Wakanomiya, Atsugi, Kanagawa, 243-0198, Japan}


\begin{abstract}
Simulating real-time dynamics under a Hamiltonian is a central goal of quantum information science.
While numerous Hamiltonian-simulation quantum algorithms have been proposed, the effects of physical noise have rarely been incorporated into their performance analysis, despite the non-negligible noise levels of quantum devices.
We present an end-to-end complexity analysis of noisy Hamiltonian simulation combined with quantum error mitigation (QEM) to answer how many circuit runs are required to reach a given target accuracy.
Because the QEM sampling overhead grows exponentially with the circuit depth while the algorithmic error decreases with it, the circuit depth becomes an optimization variable, and we derive an analytic depth-selection rule for two algorithm families.
For the order-$k$ Suzuki--Trotter formula, the optimized cost exhibits a critical error $\epsilon_c$, below which the required number of circuit runs grows exponentially.
QEM improves the noise dependence of $\epsilon_c$ from sublinear to $k$th-power scaling, an exponent improvement by a factor of $k+1$.
For randomized-LCU-based simulation, optimizing the repetition number yields a square-root improvement in the simulation-time dependence of the sampling-overhead exponent over the standard parameter choice.
We further quantify the noise-characterization cost required for error mitigation via gate set tomography and the recently proposed space-time noise inversion method, showing that the latter can significantly reduce this cost.

\end{abstract}

\maketitle
\section{Introduction} 
\label{section:introduction}
Hamiltonian simulation algorithms for the real-time evolution $e^{-iHt}$ of a Hamiltonian $H$ are at the center of the quantum computing field, as simulating quantum systems such as condensed-matter and chemical systems is a key goal for quantum computers~\cite{lloyd1996universal,berry2015simulating,low2019hamiltonian}. In addition, Hamiltonian simulation underlies many quantum algorithms, e.g., quantum phase estimation~\cite{kitaev1995quantum,aspuru2005simulated,lin2022heisenberg,wan2022randomized}, a general-dynamics quantum simulator based on the linear combination of Hamiltonian simulation~\cite{an2023linear,an2026quantum}, and linear system solvers~\cite{childs2017quantum,wang2024qubit,Chakraborty2024}.


Although noise is inherently present and cannot be ignored in practice, its impact is rarely incorporated into the performance analysis of Hamiltonian simulation algorithms, and only a few works treat it rigorously~\cite{knee2015optimal,endo2019mitigating,mohammadipour2025direct,hakkaku2025data,xu2025exponentially}. On the noise-aware side, \citet{knee2015optimal} identify a trade-off between the Trotter count and physical noise and prove the existence of an optimal Trotter count, which \citet{xu2025exponentially} recently refined, though without error mitigation. On the mitigation side, extrapolation-based quantum error mitigation (QEM) can suppress both the algorithmic error and the physical noise~\cite{endo2019mitigating,hakkaku2025data,mohammadipour2025direct,watson2025exponentially}. However, an end-to-end complexity analysis that accounts for both the QEM sampling overhead and the residual bias is available only for zero-noise extrapolation applied to the second-order Trotter formula~\cite{mohammadipour2025direct}.



In the present work, we analyze the end-to-end complexity of Hamiltonian simulation in the presence of physical noise and QEM~\cite{li2017efficient,temme2017error,endo2018practical,cai2023quantum,endo2021hybrid,aharonov2025importance}, quantified as the number of circuit runs required to reach a given target accuracy.
As the QEM method, we mainly discuss probabilistic error cancellation (PEC)~\cite{temme2017error,endo2018practical}, which inverts the noise channel based on a noise model characterized beforehand and admits analytic expressions for both the sampling overhead and the maximum bias.
We adopt the mean-squared error (MSE), denoted by $\epsilon^2$, as the performance metric~\cite{cai2021practical,cai2023quantum}, because it captures both the statistical error due to a finite number of circuit runs and the systematic error arising from the algorithmic approximation and physical noise.
The QEM sampling overhead grows exponentially with the depth, whereas the algorithmic error decreases with it.
The circuit depth thus becomes an optimization variable, and we minimize the MSE over it.

We consider two families of Hamiltonian simulation algorithms suitable for near-term and early-fault-tolerant quantum simulation, the Suzuki--Trotter product formulas~\cite{berry2007efficient,childs2021} and the randomized linear-combination-of-unitaries (RLCU)-based algorithms~\cite{childs2012hamiltonian,Chakraborty2024}. For the Trotter-based simulation, we identify two distinct accuracy regimes in the end-to-end sampling cost, a regime where the cost grows only polynomially as the target accuracy becomes small and a regime where it grows exponentially.
This transition implies the existence of a critical error $\epsilon_c$, below which further improving the accuracy becomes exponentially costly in the required number of circuit runs. Furthermore, compared with the noisy Trotter-based simulation, combining higher-order product formulas with PEC leads to a \emph{qualitative} improvement in the dependence of the critical error on the effective noise strength. It improves from a sublinear scaling to a $k$th-power scaling for an order-$k$ product formula, corresponding to an exponent improvement by a factor of $k{+}1$.  Note that this analysis is directly applicable to the qDrift algorithm because an analogous trade-off holds between the algorithmic error and the physical noise as for the first-order Trotter formula~\cite{PhysRevLett.123.070503}. We further show that the problem of the critical error can be alleviated via the Trotter extrapolation method~\cite{endo2019mitigating,watson2025exponentially}.

For the RLCU-based algorithm, we first show that the circuit depth is a random variable with a bounded mean, so the overall performance is governed by the sampling overhead rather than by an algorithmic bias.
When combined with QEM, we further optimize the repetition number, which yields a square-root improvement in the dominant time dependence appearing in the exponential sampling overhead compared with a standard (non-optimized) parameter choice, in the regime where this overhead dominates the cost.

Finally, we quantify the sampling cost required for noise characterization when implementing error mitigation, and we assess how the recently introduced space-time noise inversion (SNI) method~\cite{xie2026noise} can improve the scaling of this characterization cost in the early fault-tolerant quantum computing (FTQC) era. While the SNI method incurs a slight increase in the QEM sampling overhead, we show that a similar optimization strategy is also applicable to the SNI method, with significantly improved characterization cost.

\section{Hamiltonian simulation algorithms}

Here, we summarize the two families of Hamiltonian-simulation algorithms used throughout the paper: product-formula (Suzuki--Trotter) algorithms, presented in \cref{subsec:trotter_first_order,subsec:trotter_higher_order}, and the RLCU algorithm, presented in \cref{subsec:rlcu}.
We consider an $n$-qubit Hamiltonian $H$ and a simulation time $t>0$.
We assume throughout that the Hamiltonian admits a Pauli expansion
\begin{equation}
H=\sum_{\ell=1}^{L}\lambda_\ell P_\ell,
\end{equation}
where each $P_\ell$ is a Pauli string and $\lambda_\ell\in\mathbb{R}$.
This representation is always available for qubit Hamiltonians.
We assume that each Pauli rotation $e^{-i\lambda_\ell P_\ell\delta}$ can be implemented for any step size $\delta$.
For notational convenience, we also write
$H=\sum_{\ell=1}^{L} H_{\ell}$ by setting $H_\ell:=\lambda_\ell P_\ell$.
In this convention, $L$ denotes the number of Pauli terms in the Hamiltonian $H$.

\subsection{First-order Trotter: baseline bias}
\label{subsec:trotter_first_order}

Let $t>0$ denote the evolution time.
Our target unitary channel and expectation value are
\begin{equation}
\mathcal U_t(\rho)=e^{-iHt}\rho\,e^{+iHt},\qquad
\mu=\Tr\!\big[O\,\mathcal U_t(\rho)\big],
\end{equation}
where we assume $\|O\|_\infty\le 1$ (otherwise normalize $O$).

Let $N\in\mathbb N$ be the number of Trotter microsteps and set $\delta=t/N$.
Define the first-order (Trotter) microstep unitary and its channel as
\begin{equation}
S^{(1)}(\delta):=\prod_{\ell=1}^{L} e^{-iH_\ell\delta},\qquad
\mathcal S^{(1)}_{\delta}(\rho):=S^{(1)}(\delta)\,\rho\,S^{(1)}(\delta)^\dagger.
\end{equation}
The $N$-step Trotter channel is given by the composition
\begin{equation}
\mathcal V_N^{(1)} := (\mathcal S^{(1)}_{\delta})^{N}.
\end{equation}
We estimate $\mu$ by $\hat\mu$ whose mean is
$\E[\hat\mu]=\Tr[O\,\mathcal V_N^{(1)}(\rho)]$.

\paragraph*{Bias described by the diamond distance --}
Throughout this paper, we use the diamond distance to characterize the bias in the Hamiltonian simulation algorithms.
\begin{equation}
    \|\Phi-\Psi\|_{\diamond}
    \;:=\;
    \max_{\rho}\,
    \bigl\|
        (\Phi \otimes \mathrm{id})(\rho)
        -
        (\Psi \otimes \mathrm{id})(\rho)
    \bigr\|_{1},
\end{equation}
where $\|\cdot\|_{1}$ denotes the trace norm. We write $D(\Phi,\Psi) := \frac{1}{2}\|\Phi-\Psi\|_{\diamond}$.
 We can alternatively use the induced trace distance $
    D_{\mathrm{IT}}(\Phi,\Psi)
    \;:=\; \frac{1}{2}
    \max_{\rho}\,
    \bigl\|\Phi(\rho)-\Psi(\rho)\bigr\|_{1}$, where the maximization is over density operators $\rho$.
These distance measures obey the triangle inequality and are contractive under composition with CPTP maps:
for any CPTP map $\Lambda$,
\begin{equation}
\begin{aligned}
    \label{eq:D_contractive_rewrite}
    &D(\Lambda\circ\Phi,\Lambda\circ\Psi)\le D(\Phi,\Psi),
    \\
    &D(\Phi\circ\Lambda,\Psi\circ\Lambda)\le D(\Phi,\Psi).
\end{aligned}
\end{equation}
In particular, for any state $\rho$ and any observable $O$ with $\|O\|_\infty\le 1$,
\begin{equation}
\big|\Tr[O\mathcal A(\rho)]-\Tr[O\mathcal B(\rho)]\big|
\le 2D(\mathcal A,\mathcal B).
\label{eq:D_dominates_rewrite}
\end{equation}
Then the algorithmic bias satisfies
\begin{equation}
|\E[\hat\mu]-\mu|
\le  2D(\mathcal U_t,\mathcal V_N^{(1)}).
\label{eq:bias_by_channel_distance_rewrite}
\end{equation}

Note that, while the induced trace distance measure gives a tighter bound on the bias of the expectation values than the diamond distance, the QEM sampling overhead can be systematically bounded by using the diamond distance as we show in \cref{sec:pec}. Therefore, we mainly use the diamond distance measure.

\paragraph*{A commutator bound for the first-order step --}
We can bound the first-order Trotter error by using commutators.
In particular,~\citet{childs2021} prove a tight additive-error bound for the first-order
Trotter formula:
\begin{equation}
\big\|S^{(1)}(t)-e^{-iHt}\big\|
\;\le\;
\frac{t^{2}}{2}\sum_{\ell=1}^{L}
\Big\|\Big[\sum_{m>\ell}H_m,\;H_\ell\Big]\Big\|.
\label{eq:prop15_like_rewrite}
\end{equation}
Here $\|\cdot\|$ is the operator norm.
Applying \cref{eq:prop15_like_rewrite} to one microstep of size $\delta$ and
using $\delta=t/N$, we obtain a per-step bound
$\|S^{(1)}(\delta)-e^{-iH\delta}\|\le c_1\,\delta^2$
with the explicit prefactor
\begin{equation}
c_1:=\frac12\sum_{\ell=1}^{L}
\Big\|\Big[\sum_{m>\ell}H_m,\;H_\ell\Big]\Big\|.
\label{eq:c1_first_order_rewrite}
\end{equation}
For unitary channels, the channel distance is bounded by the operator-norm difference,
$D(\mathcal A,\mathcal B)\le\|A-B\|$,
so $D(\mathcal S^{(1)}_{\delta},\mathcal U_\delta)\le c_1\,\delta^2$.
The triangle inequality and contractivity of $D$ over $N$ steps then give
\begin{equation}
D\!\left(\mathcal U_t,\mathcal V_N^{(1)}\right)
\;\le\;
N\,c_1\,\delta^2
\;=\;
\frac{c_1\,t^2}{N}.
\label{eq:first_order_global_rewrite}
\end{equation}
We absorb the $t$-dependence into the first-order Trotter constant
\begin{equation}
\alpha_1 := c_1\,t^2,
\qquad
D\!\left(\mathcal U_t,\mathcal V_N^{(1)}\right)\le\frac{\alpha_1}{N}.
\label{eq:alpha1_rewrite}
\end{equation}
Thus, at fixed Hamiltonian decomposition, the algorithmic bias decays as $O(1/N)$
with an explicit commutator-controlled prefactor $\alpha_1$.


\subsection{Higher-order Suzuki--Trotter product formulas and commutator scaling}
\label{subsec:trotter_higher_order}

We now consider even-order Suzuki--Trotter formulas.
Write $N$ for the number of repeated \emph{microsteps} and set $\delta=t/N$.

\paragraph*{Recursive construction for order $k=2p$ for an integer $p\ge 1$ --}
We start by defining the second-order microstep channel
\begin{align}
\mathcal S^{(2)}_{\delta}(\rho)
&:=\Big(\prod_{\ell=1}^{L} e^{-iH_\ell \tfrac{\delta}{2}}\Big)
   \Big(\prod_{\ell=L}^{1} e^{-iH_\ell \tfrac{\delta}{2}}\Big)\,
   \rho \nonumber\\
&\qquad\times
   \Big(\prod_{\ell=1}^{L} e^{+iH_\ell \tfrac{\delta}{2}}\Big)
   \Big(\prod_{\ell=L}^{1} e^{+iH_\ell \tfrac{\delta}{2}}\Big),
\end{align}
and the corresponding $N$-step channel $\mathcal V_N^{(2)} := (\mathcal S^{(2)}_{\delta})^{N}$.
For even orders $k=2p\ge 4$, we recursively define the microstep channel~\cite{suzuki1990fractal,suzuki1991general}
\begin{equation}
\begin{aligned}
\mathcal S^{(k)}_{\delta}
&:=
\bigl(\mathcal S^{(k-2)}_{u_p\delta}\bigr)^{2}~
\mathcal S^{(k-2)}_{(1-4u_p)\delta}~
\bigl(\mathcal S^{(k-2)}_{u_p\delta}\bigr)^{2}, \\
\qquad
u_p&=\frac{1}{4-4^{1/(2p-1)}}.
\end{aligned}
\end{equation}
The corresponding channel is
\begin{equation}
\mathcal V_{N}^{(k)} := (\mathcal S^{(k)}_{\delta})^{N},
\end{equation}
with $\delta=t/N$.

\paragraph*{Commutator-controlled prefactors --}
Following~\citet{childs2021}, the global Trotter error for order $k$ satisfies
$D(\mathcal U_t,\mathcal V_N^{(k)})\le O(\alpha^{(k)}_{\mathrm{comm}}\,t^{k+1}/N^{k})$,
where $\alpha^{(k)}_{\mathrm{comm}}$ is the prefactor defined as
\begin{equation}
\alpha^{(k)}_{\mathrm{comm}}
:=\sum_{\ell_1,\ldots,\ell_{k+1}=1}^{L}
\big\|\big[H_{\ell_{k+1}},\ldots,[H_{\ell_2},H_{\ell_1}]\ldots\big]\big\|.
\end{equation}
Each $k$th-order microstep $\mathcal S^{(k)}_{\delta}$, however, unfolds recursively into
$\Upsilon_k = 2\times 5^{k/2-1}$ stages,
so the total number of circuit layers is $d = \Upsilon_k N$.
Thus the bound can be rewritten as
\begin{equation}
D\!\left(\mathcal U_t,\mathcal V_N^{(k)}\right)\;\le\;\frac{\alpha_k}{d^{k}},
\label{eq:trotter_d_rewrite}
\end{equation}
where $\alpha_k = O\!\left(\Upsilon_k^{k}\,\alpha^{(k)}_{\mathrm{comm}}\,t^{k+1}\right)$.
For geometrically local Hamiltonians, $\alpha^{(k)}_{\mathrm{comm}}=O(L)$,
because each term $H_\ell$ commutes with all but $O(1)$ neighbors, causing most nested
commutators to vanish~\cite{childs2021}.
Note that \cref{eq:trotter_d_rewrite} and the definition of
$\alpha^{(k)}_{\mathrm{comm}}$ encompass the first-order case as a special instance,
where we define $\Upsilon_1 = 1$,
i.e.\ $\alpha_1 = O(\alpha^{(1)}_{\mathrm{comm}}\,t^{2})$ as in~\cref{eq:alpha1_rewrite}.

\subsection{Randomized linear combination of unitaries for real-time simulation}
\label{subsec:rlcu}

Here, we review the recently introduced single-ancilla LCU algorithm for real-time simulation~\cite{Chakraborty2024}, which we refer to as the randomized LCU (RLCU) algorithm. 
Because this algorithm uses a random sampling of gates, it is convenient to introduce the normalized Hamiltonian
\begin{equation}
    \tilde H := \frac{H}{\beta} = \sum_{s=1}^{L} p_s P_s,
\end{equation}
where $\beta=\sum_{s=1}^L |\lambda_s|$ and $p_s := \frac{|\lambda_s|}{\beta}\ge 0$. Here we absorb $\mathrm{sgn}(\lambda_s)$ into the operator, replacing $P_s$ by $\mathrm{sgn}(\lambda_s)\,P_s$, so that the coefficients $p_s$ form a probability distribution while each $P_s$ remains a (signed) Pauli string. We also introduce the scaled time as $\tilde t := \beta t$.
We measure the gate count in a gate model where a Pauli rotation 
$e^{-i\theta P_s}$ is treated as one elementary operation.
We approximate the segment evolution
$e^{-iHt/r}$ by the Taylor expansion~\cite{wan2022randomized,berry2015simulating}
\begin{equation}
    S_r := \sum_{k=0}^{K} \frac{(-i t H / r)^k}{k!}
    = \sum_{k=0}^{K} \frac{(-i \tilde t \tilde H / r)^k}{k!},
\end{equation}
and later repeat this segment $r$ times to approximate $e^{-iHt}$.
Originally, a finite cutoff parameter $K$ was introduced in the Taylor expansion to truncate the infinite series.
We instead take the limit $K\to\infty$ and retain the full Taylor series.
As we explain below, sampling from the resulting infinite summation can be performed efficiently.

\paragraph*{LCU decomposition of one segment $S_r$ --}

\begin{figure}[htp] \includegraphics[width=0.42\textwidth]{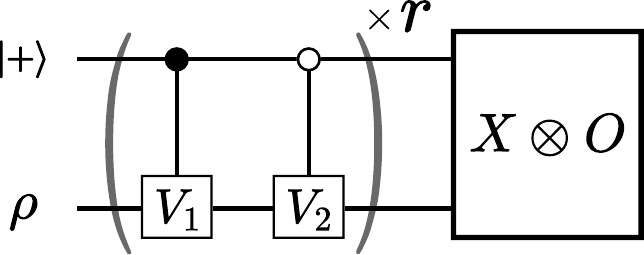} 
\caption{Quantum circuit for implementing the linear map 
$(S_r)^r = (\sum_\mu \alpha_\mu U_\mu)^r$ and measuring an observable $O$ against an input state $\rho$, 
using the RLCU algorithm. We randomly sample $V_1$ and $V_2$ according to the probability distribution $\{ 
\alpha_\mu / \|\alpha\|_{\ell_1}, U_\mu \}$, where $\|\alpha\|_{\ell_1} = \sum_\mu |\alpha_\mu|$.}
\label{fig:single-lcu}
\end{figure}

Let $\tau := \tilde t/r$. We group the Taylor-series terms by pairing even and odd orders:
\begin{equation}
\begin{aligned}
    S_r &:=
    \sum_{k=0}^{\infty}
        \frac{(- i \tau \tilde H)^k}{k!}
    \\
    &=
    \sum_{k \in \mathrm{even}}^{\infty}
        \frac{1}{k!}
        \left(- i\tau\tilde H \right)^k
        \left(
            I - \frac{i \tau \tilde H}{k+1}
        \right)
    \\
    &=
    \sum_{k \in \mathrm{even}}^{\infty}
        \frac{(- i \tau)^k}{k!}
        \sqrt{1 + \left(\frac{\tau}{k+1}\right)^2}
    \\
    &\qquad \times
    \sum_{\ell_1, \ell_2, \cdots, \ell_k, m}^{L}
        p_{\ell_1} \cdots p_{\ell_k} p_m\;
        P_{\ell_1} \cdots P_{\ell_k} e^{- i \theta_k P_m},
\end{aligned}
\end{equation}
where $\theta_{k}$ is chosen so that
\begin{equation}
    e^{-i\theta_{k} P_m}
    =
    \frac{I - i\frac{\tau}{k+1} P_m}{\sqrt{1+\left(\frac{\tau}{k+1}\right)^2}}.
\end{equation}
Thus $S_r$ is written as an LCU:
\begin{equation}
    \label{eq:S_r_def}
    S_r = \sum_{\mu} \alpha_{\mu} U_{\mu},
\end{equation}
where we may take the multi-index $\mu=(k,\ell_1,\ldots,\ell_k,m)$ with even $k$, and define
\begin{equation}
    \label{eq:U_mu_def}
    U_{(k,\boldsymbol{\ell},m)}
    :=
    (-1)^{k/2}\,P_{\ell_1}\cdots P_{\ell_k}\,e^{-i\theta_k P_m},
\end{equation}
where the global phase $(-1)^{k/2}$ (which equals $(-i)^{k}$ for even $k$) is absorbed into $U_{(k,\boldsymbol{\ell},m)}$, which therefore remains unitary and ensures $S_r=\sum_\mu\alpha_\mu U_\mu$, and the corresponding nonnegative coefficient
\begin{equation}
    \alpha_{(k,\boldsymbol{\ell},m)}
    :=
    \frac{\tau^{k}}{k!}
    \sqrt{1+\left(\frac{\tau}{k+1}\right)^2}\;
    \Bigl(\prod_{j=1}^{k} p_{\ell_j}\Bigr)p_m.
    \label{eq:alpha_explicit}
\end{equation}
The 1-norm of the LCU coefficients is then
\begin{equation}
    \|\alpha\|_{\ell_1}
    := \sum_{\mu} \alpha_{\mu}
    = \sum_{k \in \mathrm{even}}^{\infty}
    \frac{\tau^{k}}{k!}
    \sqrt{1+\left(\frac{\tau}{k+1}\right)^2},
    \label{eq:alpha_1norm}
\end{equation}
where we used $\sum_{\ell_1,\ldots,\ell_k,m}(\prod_j p_{\ell_j})p_m=1$.

\paragraph*{Sampling procedure and estimator}
Let $\mathcal{S}_r$ denote the unitary channel associated with one segment, $\mathcal{S}_r(\rho):=S_r\,\rho\,S_r^{\dagger}$, so that $(\mathcal{S}_r)^{r}(\rho)=S_r^{r}\,\rho\,(S_r^{r})^{\dagger}$.
To estimate the expectation value $\Tr[O\,(\mathcal{S}_r)^{r}(\rho)]$ for an observable $O$ and input state $\rho$,
we proceed as follows.
We initialize the joint ancilla--system state as $|{+}\rangle\langle{+}|\otimes\rho$
and apply the circuit in \cref{fig:single-lcu} for $r$ rounds.
In each round $j\in\{1,\ldots,r\}$, we independently sample two unitaries $V_{1,\mu}$ and $V_{2,\nu}$ from the distribution $\{p_\mu, U_\mu\}$, where
\begin{equation}
    p_\mu := \frac{\alpha_\mu}{\|\alpha\|_{\ell_1}},
    \label{eq:P_mu_def}
\end{equation}
with the total (algorithmic) sampling overhead being $\Gamma_{\mathrm{RLCU}} := (\Gamma_S)^{r}$.
We call $\Gamma_S := \|\alpha\|_{\ell_1}^{\,2}$ the segment sampling overhead.
Sampling from the infinite series can be carried out by first drawing the Taylor order $k$ from an even-Poisson distribution and then applying thinning to obtain a sample from $p(k,\tau)$, followed by selecting each Pauli string $P_\ell$ independently according to the discrete distribution $\{p_\ell\}$.
This random process implements the microstep channel $\tilde{\mathcal{S}}_r$ at the channel level.
For each sampled pair $(U_\mu, U_\nu)$, we define the joint ancilla--system gate
\begin{equation}
    G_{\mu\nu}
    :=
    |0\rangle\langle 0|_A \otimes U_\mu
    +
    |1\rangle\langle 1|_A \otimes U_\nu,
    \label{eq:G_def}
\end{equation}
which applies $U_\mu$ when the ancilla is in $|0\rangle$ and $U_\nu$ when it is in $|1\rangle$.
The averaged channel is
\begin{equation}
    \tilde{\mathcal{S}}_r(\sigma)
    :=
    \sum_{\mu,\nu}
    p_\mu\,p_\nu\;
    G_{\mu\nu}\,\sigma\,G_{\mu\nu}^{\dagger}
    \label{eq:tilde_S_def}
\end{equation}
for any density matrix $\sigma$ on the joint ancilla--system space.
After $r$ rounds, we measure $X\otimes O$ on the final state and scale by $\Gamma_{\mathrm{RLCU}}$.
A single ancilla qubit suffices throughout,
and measuring $X\otimes O$ followed by rescaling by $\Gamma_{\mathrm{RLCU}}$ 
gives the unbiased estimator of $\Tr[O\,(\mathcal{S}_r)^{r}(\rho)]$:
\begin{equation}
    \begin{split}
    &\Gamma_{\mathrm{RLCU}}\Tr\!\left[
        (X\otimes O) \;\;
        \tilde{\mathcal{S}}_r^{\,r}\!\left(|{+}\rangle\langle{+}|\otimes\rho\right)
    \right] \\
    &= 
    \Gamma_{\mathrm{RLCU}}\Tr
    \Bigl[ O \Bigl(\sum_\mu p_\mu U_\mu\Bigr)^r \rho 
    \Bigl(\sum_\nu p_\nu U_\nu^\dagger\Bigr)^r \Bigr]\\
    &= 
    \Tr
    \bigl[O S_r^r \rho (S_r^r)^\dagger\bigr].
    \end{split}
\end{equation}
Here, we used the definition of $\Gamma_{\mathrm{RLCU}} = \|\alpha\|_{\ell_1}^{2r}$ 
and~\cref{eq:S_r_def} to obtain the last line.
Denoting the measured outcome by $\hat{o}$, 
the unbiased estimator is
\begin{equation}
    \widehat{\mu}_r
    :=
    \Gamma_{\mathrm{RLCU}}\,\frac{1}{M}\sum_{j=1}^{M}\hat{o}_j,
    \label{eq:Sr_estimator}
\end{equation}
where $M$ is the number of circuit runs.
The rescaling by $\Gamma_{\mathrm{RLCU}}$ amplifies statistical fluctuations by a factor $\Gamma_{\mathrm{RLCU}}^2$.
Because $\mathrm{Var}(\frac{1}{M}\sum_{j=1}^{M}\hat{o}_j)=\mathrm{Var}(\hat{o})/M$
and $\| \hat{o} \|_{\infty} \leq 1$,
\begin{equation}
    \label{eq:var_Sr_estimator}
    \mathrm{Var}(\widehat{\mu}_r)
    =
    \Gamma_{\mathrm{RLCU}}^{2}\,\frac{\mathrm{Var}(\hat{o})}{M}
    \le
    \frac{\Gamma_{\mathrm{RLCU}}^{2}}{M}.
\end{equation}
\citet{wan2022randomized} and~\citet{Chakraborty2024} show that the segment $\ell_1$-norm satisfies $\|\alpha\|_{\ell_1}\le\exp(\tilde t^{2}/r^{2})$, so that the segment sampling overhead satisfies
$\Gamma_S=\|\alpha\|_{\ell_1}^2 \leq \exp\!\left(\frac{2\tilde t^{2}}{r^{2}}\right),$
and therefore the total sampling overhead is bounded by
\begin{equation}
    \label{eq:Gamma_RLCU_bound}
    \Gamma_{\mathrm{RLCU}}
    \le
    \exp\!\left(\frac{2\tilde t^{2}}{r}\right).
\end{equation}
\paragraph*{Expected number of Clifford gates and variance --}
An important feature of the RLCU algorithm is that the number of Clifford gates per segment is a random variable determined by the sampled Taylor order $k$.
We now bound the expected number of Clifford gates and its variance, both of which will be used in the analysis of the QEM sampling overhead in later sections.

\begin{theorem}[Average and variance of the number of Clifford gates]
    \label{thm:avg_gate_count}
Let $\tau := \tilde t/r = \beta t/r$,
and let $p(k,\tau)$ denote the probability that the RLCU algorithm samples Taylor order $k$ in a single segment.
Here, $\E_{p}$ and $\mathrm{Var}_{p}$ denote the expectation and variance with respect to $p(k,\tau)$.
Then the mean and variance of the number of Clifford gates satisfy
\begin{align}
    \E_{p}[k]
    &\le
        \tau\,\tanh(\tau) \leq \tau^2,
        \label{eq:mean_k_bound}\\
    \mathrm{Var}_{p}(k)
    &:=
    \E_{p}[k^2] - \E_{p}[k]^2
    \le
        \tau^2 + \tau\,\tanh(\tau) \leq 2\tau^2.
        \label{eq:var_k_bound}
\end{align}
\end{theorem}
The proof is given in~\cref{app:proof_avg_gate_count}.

\section{Optimizing the circuit depth in the presence of physical noise for Trotter-based algorithms and RLCU-based algorithms}
\label{sec:opt_layers_noise_st_rlcu}

In a noisy quantum device, increasing the circuit depth reduces the algorithmic approximation error, which we simply call the algorithmic error, but simultaneously accumulates physical noise.
This creates a fundamental trade-off that determines the optimal depth for each algorithm.
In this section, we analyze this trade-off for both the Suzuki--Trotter and RLCU algorithms, and identify the optimal depth that minimizes the mean-squared error (MSE).
\subsection{Suzuki--Trotter-based algorithm}
\label{subsec:st_baseline_noise}

We first analyze the Suzuki--Trotter algorithm under physical noise.
We quantify this trade-off by bounding the total distance between
the ideal evolution and the noisy Trotter channel,
using the properties of the diamond distance defined in \cref{subsec:trotter_first_order},
modeling the noisy channel of each Trotter microstep as
\begin{equation}
\mathcal E_{\delta}^{\mathrm{noisy}} = \prod_{\ell=1}^{\Upsilon_k L} \mathcal{N}_\ell \mathcal{U}^\ell_{\delta},
\end{equation}
where $\mathcal{U}_\delta^\ell$ denotes the unitary channel
corresponding to the $\ell$-th elementary operation (i.e., the Pauli rotation gate) of the microstep $\mathcal{S}^{(k)}_{\delta}$
 and $\mathcal{N}_\ell$ is the noise process associated with the $\ell$-th elementary operation.
Note that one microstep contains $\Upsilon_k L$ elementary gates, since the Hamiltonian has $L$ Pauli terms.
The distance between the noise-free ideal channel 
and the noisy channel is
\begin{equation}
\begin{aligned}
    \label{eq:per-step-noise}
    &D(\mathcal{U}_{\delta}, \mathcal{E}_{\delta}^{\mathrm{noisy}}) \\ 
    &\le D(\mathcal{U}_{\delta}, \mathcal{S}^{(k)}_{\delta})
    + D(\mathcal{S}^{(k)}_{\delta}, \mathcal{E}_\delta^{\mathrm{noisy}}) \\
    &\le D(\mathcal{U}_{\delta}, \mathcal{S}^{(k)}_{\delta})
    + \sum_{\ell=1}^{\Upsilon_k L} D(\mathcal{N}_\ell \circ \mathcal{U}^\ell_{\delta}, \mathcal{U}^\ell_{\delta}) \\
    &= D(\mathcal{U}_{\delta}, \mathcal{S}^{(k)}_{\delta})
    + \sum_{\ell=1}^{\Upsilon_k L} D(\mathcal{N}_\ell, \mathcal{I}),
\end{aligned}
\end{equation}
where $\mathcal{S}^{(k)}_{\delta}$ denotes the channel for one microstep of the order-$k$ Trotter formula and
we use the triangle inequality and the contractivity~\cref{eq:D_contractive_rewrite}.
Therefore, for the $k$th-order Suzuki--Trotter algorithm and a normalized observable $O$ with $\|O\|_\infty\le1$, the bias of the expectation value is bounded by
\begin{align}
    \label{eq:total-bias}
    &\big|\Tr[O\,\mathcal{U}_t(\rho_{\rm in})]- \Tr[O\,(\mathcal{E}_\delta^{\mathrm{noisy}})^N(\rho_{\rm in})] \big| \nonumber \\
    &\le 2D(\mathcal{U}_t, (\mathcal{E}_\delta^{\mathrm{noisy}})^N)  
    \le \frac{\alpha_k}{d^k} + Ld\gamma,
\end{align}
with 
\begin{equation}
  \label{eq:gamma}
  \gamma:= \max_\ell 2D(\mathcal{N}_\ell, \mathcal{I})=\max_\ell\|\mathcal{N}_\ell-\mathcal{I}\|_\diamond,
\end{equation}
which we call the per-gate error rate. The algorithmic contribution follows from \cref{eq:trotter_d_rewrite}, and the order-one factor arising when the channel distance is interpreted as a bias via \cref{eq:D_dominates_rewrite} is absorbed into $\alpha_k$. The noise contribution counts all $\Upsilon_k L$ gates in each of the $N$ microsteps, giving $N\,\Upsilon_k L\,\gamma = Ld\gamma$ with $d = \Upsilon_k N$.

We evaluate the performance of the noisy Hamiltonian simulation using the mean-squared error (MSE). 
The MSE serves as a unified error metric that captures both the systematic error and the statistical error in a single quantity.
The squared bias represents a consistent deviation from the ideal value (e.g., from algorithmic approximation and physical noise),
while the variance represents random fluctuations due to a finite number of circuit runs.
See~\cref{app:mse} for the basic properties of MSE and a simple conservative confidence-interval interpretation.
The MSE is defined as

\begin{equation}
    \label{eq:mse}
    \epsilon^2 = \mathrm{Var}(\hat{\mu}) + \mathrm{Bias}(\hat{\mu})^2.
\end{equation}

We call $\epsilon$ the error or accuracy (depending on the context).
Using~\cref{eq:total-bias} and assuming that we use $M$ circuit runs to estimate the expectation value, we obtain the bound:
\begin{equation}
    \label{eq:non-qem-mse}
    \epsilon^2 \le \left(\frac{\alpha_k}{d^k} + Ld\gamma \right)^2 + \frac{1}{M}.
\end{equation}
We recall that the observable is normalized so that $\|O\|_\infty \le 1$, so the variance scales as $O(1/M)$.
We can minimize the MSE in~\cref{eq:non-qem-mse} by optimizing $d$ and
obtain the optimized error bound
\begin{equation}
\epsilon^{2\star}
\;=\;
\big(C_k
\alpha_k^{1/(k+1)}\,
\left(L\gamma\right)^{k/(k+1)}\big)^2
\;+\;\frac{1}{M}.
\label{eq:mse-noqem-optimized}
\end{equation}
Here, $C_k = (k^{1/(k+1)}+k^{-k/(k+1)})$ is a constant that depends on the order $k$.
Because the first term in \cref{eq:mse-noqem-optimized} is independent of $M$, the MSE cannot be reduced below it by increasing the number of circuit runs.
We refer to the first term in~\cref{eq:mse-noqem-optimized} as the \emph{error bound} for Suzuki--Trotter simulation without PEC\@
and denote it as $\epsilon_b$.
Using the scaling $\alpha_k = O(t^{k+1}L)$ for geometrically local Hamiltonians,
the error bound scales as $O\!\left(tL(\gamma)^{k/(k+1)}\right)$.
Interestingly, without any error mitigation, the error bound scales linearly in both $t$ and $L$, while the exponent of $\gamma$ satisfies $k/(k+1) < 1$ for any finite $k$, so the noise dependence remains sublinear in $\gamma$.

\subsection{RLCU-based algorithm under noise (baseline)}
\label{subsec:rlcu_baseline_noise}

We next analyze the RLCU algorithm under physical noise.
In contrast to product formulas, in the $K\to\infty$ limit the ideal segment operator satisfies
$S_r=e^{-iHt/r}$ and hence $(S_r)^r=e^{-iHt}$.
Therefore, in the absence of physical noise, the estimator $\widehat{\mu}_r$ in \cref{eq:Sr_estimator}
is unbiased for $\mu=\Tr[O\,\mathcal U_t(\rho)]$, 
at the cost of the sampling overhead $\Gamma_{\mathrm{RLCU}}$.
With physical noise, the implemented channel becomes random because the circuit depth per segment is
a random variable determined by the sampled Taylor order.

To correctly evaluate the bias, we first express the RLCU process using~\cref{eq:tilde_S_def}.
The channel corresponding to the gate $G_{\mu\nu}$ in~\cref{eq:G_def} is denoted by
$\mathcal{G}_{\mu\nu}(\sigma) := G_{\mu\nu}\,\sigma\,G_{\mu\nu}^{\dagger}$.
We write this channel as
\begin{equation}
    \mathcal{G}_{\mu\nu}
    =
    \mathcal{G}_{1,\mu}\,\mathcal{G}_{2,\nu}
    =
    \Bigl(\prod_{\ell=1}^{L_\mu} \mathcal{G}_\mu^\ell\Bigr)
    \Bigl(\prod_{\ell=1}^{L_\nu} \mathcal{G}_\nu^\ell\Bigr),
\end{equation}
where $\mathcal{G}_{1,\mu}$ and $\mathcal{G}_{2,\nu}$ are the channels corresponding to the two control unitaries in~\cref{fig:single-lcu}, and $\mathcal{G}_\mu^\ell$ 
denotes the $\ell$-th elementary gate when 
the realization $\mathcal{G}_{1,\mu}$ is decomposed into elementary operations.
In our gate model, $L_\mu = 1 + k_\mu$, where $k_\mu$ is the Taylor order for the sampled $\mu$. Recall from~\cref{eq:U_mu_def} that each realization $U_\mu$ consists of $k_\mu$ Pauli (Clifford) gates and a single Pauli-rotation gate (non-Clifford).
The noisy channel for $\mathcal{G}_{1,\mu}$ is
\begin{equation}
    \mathcal{E}_\mu^{\mathrm{noisy}}
    :=
    \prod_{\ell=1}^{L_\mu} \mathcal{N}_\ell \circ \mathcal{G}_\mu^\ell,
\end{equation}
where $\mathcal{N}_\ell$ is the noise process associated with the $\ell$th elementary operation.
The distance between the ideal channel $\mathcal{G}_{1,\mu}$ and the noisy channel satisfies
\begin{equation}
    D(\mathcal{G}_{1,\mu}, \mathcal{E}_\mu^{\mathrm{noisy}})
    \le
    \sum_{\ell=1}^{L_\mu} D(\mathcal{N}_\ell, \mathcal{I})
    \le
    \frac{\gamma + \gamma_c\,k_\mu}{2},
\end{equation}
where $\gamma:=\max 2D(\mathcal{N},\mathcal{I})$ over non-Clifford gates and $\gamma_c$ is the analogous rate over Clifford gates (both diamond-norm per-gate error rates), and $k_\mu$ is the Taylor order for the realization $\mu$.
Taking the expectation over the algorithmic randomness and using $\E[k] \le (\beta t/r)^2$ from~\cref{thm:avg_gate_count}, we obtain
$\E[D(\mathcal{G}_{1,\mu}, \mathcal{E}_\mu^{\mathrm{noisy}})] \le \frac{1}{2}\bigl(\gamma + \gamma_c\,(\beta t/r)^2\bigr)$ per control unitary.
Therefore, denoting the noisy stroboscopic channel of $\tilde{\mathcal{S}}_r$ by $\tilde{\mathcal{E}}_r^{\mathrm{noisy}}$, the distance between the ideal and noisy segment channels satisfies
\begin{equation}
    \label{eq:per-segment-distance}
    D(\tilde{\mathcal{S}}_r, \tilde{\mathcal{E}}_r^{\mathrm{noisy}})
    \le
    \gamma + \gamma_c\,\frac{(\beta t)^2}{r^2}.
\end{equation}
Applied over the $r$ segments, the channel distance accumulates to $r$ times the per-segment distance. Because $|\Tr[O\,\cdot]|\le 2D$ for $\|O\|_\infty\le1$ and the RLCU estimator is rescaled  by $\Gamma_{\mathrm{RLCU}}$ (see~\cref{eq:Sr_estimator}), the bias introduced by noise becomes
\begin{equation}
    \label{eq:bias_rlcu_bound}
\mathrm{Bias}(\widehat{\mu}_r)
=
\bigl|\E[\widehat{\mu}_r]-\mu\bigr|
\le
\Gamma_{\mathrm{RLCU}}\cdot 2\Bigl(\gamma r + \gamma_c\,\frac{(\beta t)^2}{r}\Bigr),
\end{equation}
where $\Gamma_{\mathrm{RLCU}} \le \exp\!\bigl(2(\beta t)^2/r \bigr)$.
Combining the bias bound with the variance bound in~\cref{eq:var_Sr_estimator}, we obtain
\begin{equation}
\epsilon^2
\le
\frac{\exp\!\bigl(4(\beta t)^2/r\bigr)}{M}
+
\Bigl(2 e^{2(\beta t)^2 / r}\cdot \bigl(\gamma r + \gamma_c\,\frac{(\beta t)^2}{r}\bigr)\Bigr)^{2}.
\end{equation}

The MSE does not admit a closed-form minimizer in $r$.
However, 
because the variance term can be reduced by increasing $M$ but the bias term cannot,
we can still obtain the \emph{error bound} for the RLCU algorithm.
The leading contribution to the bias scales as 
$b(r) \sim e^{2(\beta t)^2 / r}\cdot \bigl(\gamma r + \gamma_c\,\frac{(\beta t)^2}{r}\bigr)$.
Optimizing $b(r)$ with respect to $r$, the stationary condition $\partial_r b = 0$ yields
$r^3 - 2(\beta t)^2 r^2 - (\gamma_c/\gamma)(\beta t)^2 r - 2(\gamma_c/\gamma)(\beta t)^4 = 0$.
In the regime $\gamma_c \ll \gamma$, the $(\gamma_c/\gamma)$ terms are negligible, 
and hence $r^* \approx 2(\beta t)^2$.
At this choice, the error bound scales as $\epsilon_b \sim O(\gamma (\beta t)^2)$, 
i.e.\ $O(\gamma)$ with respect to the per-gate error $\gamma$
and quadratically with respect to the scaled time $\tilde t = \beta t$ (which, for local Hamiltonians with $\beta \propto L$, corresponds to the space-time volume $Lt$).

\section{Probabilistic error cancellation}
\label{sec:pec}
The analysis in the previous section reveals fundamental limits on noisy Hamiltonian simulation: for Trotter algorithms, 
the error bound $\epsilon_b$ cannot be reduced below a noise-dependent floor by increasing the number of circuit runs alone. 
To overcome these limitations, we now review the probabilistic error cancellation (PEC) method~\cite{temme2017error,endo2018practical}, 
a quantum error mitigation technique that suppresses physical noise at the cost of increased sampling overhead.

Suppose that a noisy quantum circuit can be described by
$ \prod_{k=1}^{N_G} \mathcal{N}_k \mathcal{U}_k $,
where $\mathcal{U}_k$ denotes the $k$th ideal gate channel, $\mathcal{N}_k$ its associated noise, and $N_G$ is the total number of gates.
In PEC, the noise model is characterized in advance.
Based on this information, we construct an approximate inverse map via a quasiprobability decomposition
$\mathcal{N}_k^{-1}= \sum_{i_k} q_{i_k} \mathcal{B}_{i_k}$,
where $q_{i_k} \in \mathbb{R}$ and $\mathcal{B}_{i_k}$ are implementable basis operations.
Then we can represent the ideal channel $\mathcal{U}_{\rm id} = \prod_{k=1}^{N_G} \mathcal{U}_k$ as
\begin{equation}
\begin{aligned}
\mathcal{U}_{\rm id}&= \prod_{k=1}^{N_G} \sum_{i_k} q_{i_k} \mathcal{B}_{i_k} \mathcal{N}_k \mathcal{U}_k \\
&= \Gamma  \sum_{i_1, i_2, \ldots, i_{N_G}} \prod_{k=1}^{N_G} p_{i_{k}} \mathrm{sgn} (q_{i_k}) \mathcal{B}_{i_k} \mathcal{N}_k \mathcal{U}_k,
\end{aligned}
\end{equation}
with $\Gamma= \prod_{k=1}^{N_G} \Gamma_k$, $\Gamma_k= \sum_{i_k} |q_{i_k}|$ and $p_{i_{k}}= |q_{i_k}|/\Gamma_k$.

For an observable $O$ and input state $\rho$, 
the error-canceled expectation value is given by
\begin{equation}
    \begin{split}
    &\Tr[O \mathcal{U}_{\rm id}(\rho)] 
    =\\ 
    &\Gamma \sum_{i_1, i_2, \ldots, i_{N_G}} 
 \prod_{k=1}^{N_G} p_{i_{k}} \mathrm{sgn} (q_{i_k})\Tr\left[\prod_{k=1}^{N_G} 
 \mathcal{B}_{i_k} \mathcal{N}_k \mathcal{U}_k(\rho) O \right]. 
    \end{split}
\end{equation}
Then, the unbiased expectation value can be obtained as follows. We randomly generate the basis operations for QEM $\mathcal{B}_{i_k}$ with the probability $p_{i_k}$, and measure the observable $O$, obtaining the random variable $\hat{\nu}=\prod_{k=1}^{N_G}  \mathrm{sgn} (q_{i_k}) \hat{o}$ for the measured outcome $\hat{o}$. Finally, we repeat this procedure, with $\Gamma \langle \hat{\nu} \rangle$ being the unbiased estimator of the noiseless expectation value. Because we need to amplify the random variable with $\Gamma$, the variance scales with $\Gamma^2$, which scales exponentially with the number of gates and the error rate; such an exponential sampling overhead is known to be fundamental to quantum error mitigation~\cite{takagi2022fundamental,tsubouchi2023universal,takagi2023universal,quek2024exponentially}.

We now discuss the relationship between the error rate characterized by the distance measure and the PEC sampling overhead. For a gate with error rate $\gamma_k=  2D(\mathcal{N}_k, \mathcal{I})=\|\mathcal{N}_k-\mathcal{I}\|_\diamond$, the optimal sampling overhead for implementing \(\mathcal{N}_k^{-1}\) with the completely-positive trace non-increasing basis operations is given by~\cite{regula2021operational,jiang2021physical}
\begin{equation}
    \Gamma^{\mathrm{opt}}_k
    =
    \|\mathcal{N}_k^{-1}\|_{\diamond}.
\end{equation}
Moreover, denoting $\Delta_k = \mathcal{I}-\mathcal{N}_k$ and because \(\|\Delta_k\|_{\diamond}=\gamma_k\), we have
\begin{equation}
\begin{aligned}
\|\mathcal{N}_k^{-1}\|_{\diamond}&= \| (\mathcal{I}-\Delta_k)^{-1}\|_{\diamond} \\
&=\|\sum_{m=0} \Delta_k^m \|_{\diamond}\le
    \frac{1}{1-\gamma_k},
\end{aligned}
\end{equation}
where we use the triangle inequality as $\| \sum_{m=0} \Delta_k^m \| \leq \sum_{m=0} \|\Delta_k \|^m$ in the last line.
Therefore,
\begin{equation}
    \Gamma^{\mathrm{opt}}_k
    \le
    \frac{1}{1-\gamma_k}.
\end{equation}
Because $\gamma_k\le\gamma$, for small $\gamma$, this further implies
\begin{equation}
\Gamma^{\mathrm{opt}}_k
    \le
    \frac{1}{1-\gamma}
=1+\gamma+O(\gamma^2).
\end{equation}

While the PEC sampling overhead depends on the concrete choice of the basis operation, we assume that the per-gate PEC overhead is represented as $\Gamma_k \leq 1+\gamma'$,
where $\gamma'\propto\gamma$ with a proportionality constant that depends on the noise model
and the particular quasi-probability decomposition. 
We then obtain
\begin{equation}
    \label{eq:PEC-overhead}
    \Gamma = {(1+\gamma')}^{N_G} \simeq e^{N_G \gamma'}.
\end{equation}

\section{Optimizing the simulation algorithms with PEC}
\label{sec:opt_qem}
Having established the PEC formalism, we now combine it with the Hamiltonian simulation algorithms from~\cref{sec:opt_layers_noise_st_rlcu}.
Because PEC incurs a sampling overhead that grows exponentially with the circuit depth, the optimal depth must now balance the algorithmic error, the physical noise, and the PEC overhead.
In the Trotter case, the gate count is $N_G = d L$, so the PEC overhead scales as $\Gamma \simeq e^{Ld\gamma'}$.
Note that we ignore the bias due to the incomplete characterization of noise in this section. We later discuss the required characterization cost to achieve the target accuracy.   

\subsection{Optimizing the depth of the Trotter-based algorithm under QEM}
PEC removes the noise-induced bias at the cost of a sampling overhead $\Gamma = e^{Ld\gamma'}$. The MSE then reads
\begin{equation}
\label{eq:epsilon_d_m}
   \epsilon(d, M)^2 = \left(\frac{\alpha_k}{d^k} \right)^2 + \frac{e^{2Ld\gamma'}}{M}.
\end{equation}
The stationary condition $\partial_{d}\epsilon(d, M) = 0$ yields
\begin{equation}
\label{eq:optimal_M}
M
=
\frac{L\gamma'}{k\,\alpha_k^2}\,e^{2Ld\gamma'}d^{2k+1}.
\end{equation}
Substituting this $M$ into~\cref{eq:epsilon_d_m}, the optimal circuit depth $d(\epsilon)$ for a given accuracy can be obtained by solving
\begin{equation}
\label{eq:optimal_d}
    \epsilon^2 = \frac{\alpha_k^2}{d^{2k}}\left(1+\frac{k}{\gamma' Ld} \right).
\end{equation}
To derive $M(\epsilon)$, we first note that the two terms inside the parentheses of
~\cref{eq:optimal_d} have distinct origins.
The ``$1$'' comes from the algorithmic error in the Suzuki--Trotter decomposition,
and $k/(\gamma' Ld)$ from the PEC sampling overhead.
The crossover $k/(\gamma' L d)=1$ occurs at $d = k/(\gamma' L)$, which sets the
critical-error scale $\epsilon_c := \alpha_k\!\left(L\gamma'/k\right)^{k}$.
For $\epsilon \ll \epsilon_c$, the required circuit depth satisfies $d \gg k/(\gamma' L)$ and the algorithmic error dominates.
For $\epsilon \gg \epsilon_c$, we have $d \ll k/(\gamma' L)$ and the PEC sampling overhead dominates.
By combining the above observations with~\cref{eq:optimal_M} and~\cref{eq:optimal_d}, we can obtain the following asymptotic form of $M(\epsilon)$:
\begin{equation}
    \label{eq:M_trotter}
M(\epsilon) \simeq 
\begin{cases}
    \displaystyle
    \frac{1}{\epsilon^2}\left(\frac{\epsilon_c}{\epsilon}\right)^{1/k}\exp\!\left[2k\left(\frac{\epsilon_c}{\epsilon}\right)^{1/k}\right] & \text{for } \epsilon \ll \epsilon_c , \\[6pt]
    \displaystyle
    \frac{1}{\epsilon_c^2}\,e^{2k} & \text{for } \epsilon = \epsilon_c , \\[6pt]
    \displaystyle
    \frac{1}{\epsilon^2}\left(1 + 2k\left(\frac{\epsilon_c}{\epsilon}\right)^{2/(2k+1)}\right) & \text{for } \epsilon \gg \epsilon_c.
\end{cases}
\end{equation}
Here, we treat $k = O(1)$ as a constant.
In the third case we used the linear approximation
$\exp\!\bigl[2k(\epsilon_c/\epsilon)^{2/(2k+1)}\bigr]\simeq 1+2k(\epsilon_c/\epsilon)^{2/(2k+1)}$.
This result shows that targeting an accuracy below the threshold $\epsilon_c$ requires an exponentially large sampling cost.
Accordingly, we refer to $\epsilon_c$ as the \emph{critical error} of the Trotter-based algorithm under PEC\@.
Like the error bound $\epsilon_b$, it characterizes a threshold on the attainable accuracy. Unlike $\epsilon_b$, however, $\epsilon_c$ is attained at $d = k/(\gamma' L)$.
This choice corresponds to a sampling overhead $\Gamma^2 = e^{2k}$.

Compared with the without-PEC error bound $\epsilon_b$ of~\cref{eq:mse-noqem-optimized}, the with-PEC critical error $\epsilon_c$ depends more strongly on the effective noise strength: at the level of the $L\gamma$-exponent it improves from $k/(k+1)$ to $k$ (using $\gamma' \propto \gamma$), a factor-$(k+1)$ gain.

The exponential barrier in~\cref{eq:M_trotter} is unavoidable when we use direct Trotter decompositions.
Indeed, even without considering the statistical error, 
to reach a target accuracy $\epsilon$, \cref{eq:trotter_d_rewrite} requires
$d = (\alpha_k/\epsilon)^{1/k}$ Trotter layers.
Substituting this scaling into the PEC sampling overhead $\Gamma = e^{Ld\gamma'}$ gives
$\Gamma = e^{L\gamma'(\alpha_k/\epsilon)^{1/k}} \simeq e^{k(\epsilon_c/\epsilon)^{1/k}}$.
Because this exponential overhead dominates at large $d$,
the same exponential dependence on $\epsilon$ persists
in the full expression~\cref{eq:M_trotter} whenever $\epsilon \ll \epsilon_c$,
which hinders the accuracy improvement beyond the critical error.
However, this problem can be mitigated by using the Trotter-extrapolation method and the RLCU method. See \cref{Sec:extrapolation} for the scaling improvement via the Trotter extrapolation method. In particular, the RLCU algorithm has a gate count that is independent of the target accuracy,
so the corresponding $M(\epsilon)$ does not exhibit an exponential dependence on $\epsilon$.

Because the qDrift algorithm~\cite{PhysRevLett.123.070503} exhibits a similar scaling as the first-order Trotter formula, namely that the algorithmic error is inversely proportional to the circuit depth, we can analogously optimize the circuit depth and determine the critical error threshold.

\subsection{RLCU-based algorithm with PEC}

We now turn to the RLCU algorithm combined with PEC\@.
Because the RLCU estimator is unbiased, the end-to-end performance is governed by the sampling overhead, which we optimize with respect to the repetition number $r$.
Since the circuit depth of RLCU is a random variable determined by the sampled Taylor order $k$, the PEC sampling overhead must be evaluated after averaging over this depth distribution.
The following theorem bounds the expected sampling overhead of RLCU combined with PEC\@.

\begin{theorem}
\label{thm:avg_pec_overhead}
The MSE of this algorithm is bounded by
\begin{equation}
    \label{eq:phi_r_def}
    \epsilon^2
    \le
    \frac{1}{M}
    \exp\!\left(
        \frac{4\tilde t^2}{r} + 4\gamma' r
        + \tilde t\bigl(e^{4\gamma_c' }-1\bigr)
    \right).
\end{equation}
Here, $\gamma'$ and $\gamma_c'$ denote the per-gate PEC-overhead rates for non-Clifford and Clifford operations, respectively, satisfying $\gamma'\propto\gamma$ and $\gamma_c'\propto\gamma_c$ in terms of the corresponding error rates, and typically $\gamma_c' \ll \gamma'$.
\end{theorem}
The proof is given in~\cref{app:proof_avg_pec_overhead}.
Because $\gamma_c' \ll 1$, we may approximate $e^{4\gamma_c'} - 1 \simeq 4\gamma_c'$.

As in the Trotter case, we express the result in terms of the number of circuit runs $M(\epsilon)$ required to achieve a target accuracy $\epsilon$.
Inverting the mean-squared error bound in \cref{eq:phi_r_def} under this approximation gives
\begin{equation}
    \label{eq:epsilon_r_M_approx}
    M(\epsilon, r)
    =
    \frac{1}{\epsilon^2}
    \exp\!\left(
        \frac{4\tilde t^2}{r}
        + 4\gamma'  r
        + 4\gamma_c'  \tilde t
    \right).
\end{equation}
Now we consider the optimization of the repetition number $r$.
\citet{Chakraborty2024} adopted $r = \tilde t^2 = \beta^2 t^2$ so that
the sampling overhead due to RLCU $\Gamma_{\mathrm{RLCU}} = e^{2(\beta t)^2 / r}$ is constant.
With this strategy, the required number of circuit runs becomes
\begin{equation}
    \label{eq:epsilon_r_M_approx_chakraborty}
    M(\epsilon)
    =
    \frac{1}{\epsilon^2}
    \exp\!\left(
        4
        + 4\gamma'  \tilde t^2
        + 4\gamma_c'  \tilde t
    \right).
\end{equation}
When the term proportional to $\tilde t^2$ dominates, the leading contribution in the exponent is $O(\gamma (\beta t)^2)$.
We instead optimize $M(\epsilon, r)$ in~\cref{eq:epsilon_r_M_approx} with respect to the repetition number $r$.
By balancing the terms $4\tilde t^2/r$ and $4\gamma'  r$ in the exponent, the optimal choice is $r^* = \tilde t/\sqrt{\gamma'}$.
With this choice, the required number of circuit runs reduces to
\begin{equation}
    \label{eq:epsilon_r_M_star}
    M(\epsilon)^*
    =
    \frac{1}{\epsilon^2}
    \exp\!\left(
    8\sqrt{\gamma' } \tilde t + 4\gamma_c'  \tilde t
    \right).
\end{equation}

\begin{figure}[!tbp]
\centering
\includegraphics[width=0.46\textwidth]{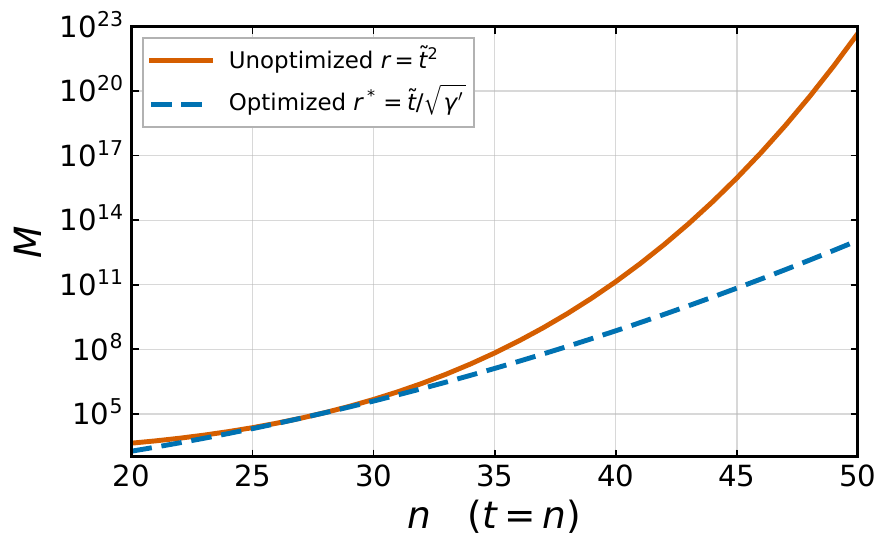}
\caption{%
Required number of circuit runs $M$ for the RLCU algorithm combined with PEC,
versus the number of sites $n$ of the XYZ chain at $t=n$,
with per-gate PEC-overhead rate $\gamma'=2 \times 10^{-7}$ (to be consistent with~\cref{fig:M_vs_n_eps}) and target accuracy $\epsilon=0.1$.
The unoptimized repetition $r=\tilde t^{2}$ of~\cref{eq:epsilon_r_M_approx_chakraborty} (solid)
is compared with the optimized $r^{*}=\tilde t/\sqrt{\gamma'}$ of~\cref{eq:epsilon_r_M_star} (dashed).
}
\label{fig:rlcu_opt}
\end{figure}

In this case, the dominant contribution in the exponent is $O(\sqrt{\gamma'}\,\beta t)$ when the square-root term dominates.
This result demonstrates a square-root improvement in the exponential scaling compared to the unoptimized choice of $r = \tilde t^2$ in~\cref{eq:epsilon_r_M_approx_chakraborty}. 
In other words, optimizing the repetition number reduces the dominant contribution in the PEC overhead from
$O(\gamma (\beta t)^2)$ to $O(\sqrt{\gamma}\,\beta t)$, yielding a mitigation of the exponential sampling cost.
\cref{fig:rlcu_opt} compares the two repetition choices for the XYZ chain at $t=n$.
The optimized repetition $r^{*}$ always yields a smaller sampling cost $M$ than $r=\tilde t^{2}$, by many orders of magnitude, so optimization is strictly preferable. This reduction, however, comes at the price of a larger circuit depth for $\sqrt{\gamma'}\,\beta t < 1$. The best choice in practice therefore depends on the total runtime, the total sampling budget, and the GST characterization cost, as discussed in~\cref{sec:gst}.

\subsection{Characteristics of the two algorithms}
\label{subsec:threshold_comparison}

The two algorithms differ qualitatively in their noise-limited accuracy.
The exponent of $M(\epsilon)$ in~\cref{eq:epsilon_r_M_approx} is independent of the target accuracy $\epsilon$, so the RLCU sampling cost grows only polynomially in $1/\epsilon$ and no critical error arises.
The Trotter sampling cost, in contrast, grows exponentially once $\epsilon$ falls below the critical error $\epsilon_c$.
\Cref{tab:error_summary} summarizes the error bound $\epsilon_b$ and the critical error $\epsilon_c$ for both algorithms.


\begin{table}[t]
\centering
\caption{%
  Summary of noise-limited accuracy thresholds for Hamiltonian simulation
  algorithms with and without PEC\@.
  We use the assumption that $\gamma' \propto \gamma$ to simplify the notation in the scaling analysis.
}
\label{tab:error_summary}
\begin{tabular}{lcc}
\toprule
 & Without PEC & With PEC \\
Algorithm & Error bound $\epsilon_b$ & Critical error $\epsilon_c$ \\
\midrule
Suzuki--Trotter
  & $O\!\left(\alpha_k^{1/(k+1)}\!\left(L\gamma \right)^{k/(k+1)}\right)$
  & $O\!\left(\alpha_k\!\left(L\gamma\right)^{\!k}\right)$ \\[12pt]
Randomized LCU
  & $O\!\left(\gamma(\beta t)^2\right)$
  & N/A \\
\bottomrule
\end{tabular}
\end{table}

\begin{figure*}[!tbp]
\centering
\begin{subfigure}[t]{0.48\textwidth}
\includegraphics[width=\textwidth]{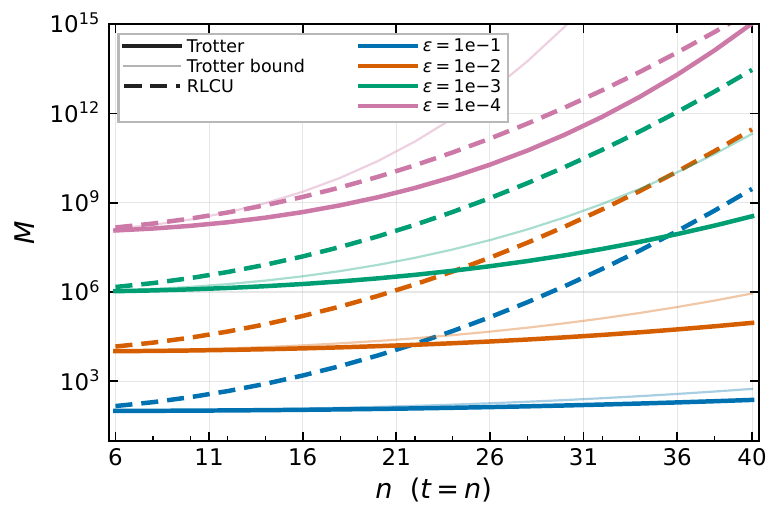}
\caption{$M$ versus $n$.}
\label{fig:M_vs_n}
\end{subfigure}
\hfill
\begin{subfigure}[t]{0.48\textwidth}
\includegraphics[width=\textwidth]{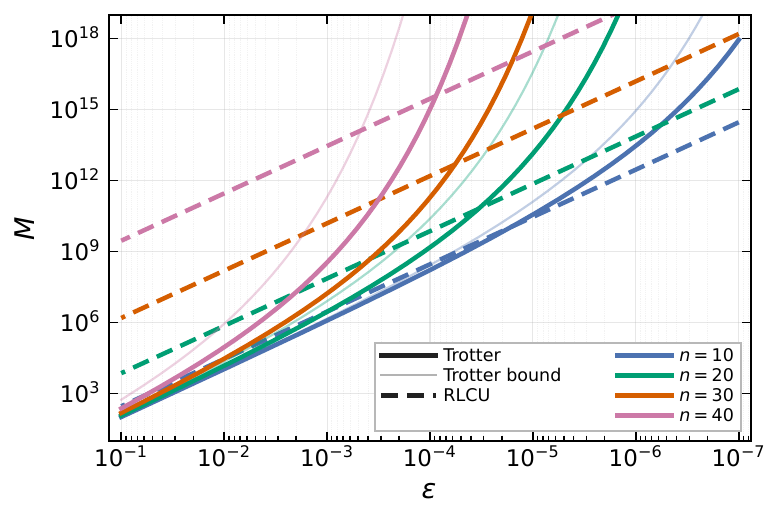}
\caption{$M$ versus $\epsilon$.}
\label{fig:M_vs_eps}
\end{subfigure}
\caption{%
Required number of circuit runs $M$ for error-mitigated Hamiltonian simulation
of a translationally invariant and periodic one-dimensional XYZ chain on $n$ sites at simulation time $t=n$,
with per-gate PEC-overhead rate $\gamma'=2\times10^{-7}$ and second-order Suzuki--Trotter ($k=2$).
Each data point is the median over $10$ random instances of nearest-neighbor couplings
normalized so that $\|H\|_{\ell_1}=L = 3n$.
``Trotter'' uses $\alpha_k$ extracted from the diamond distance
$D(\mathcal{U}_t,\mathcal{V}_N^{(2)})$ computed for $n\le 12$
and fit-extrapolated to larger $n$ (see~\cref{app:alpha_fit}).
``Trotter bound'' uses the commutator prefactor of \cref{eq:trotter_d_rewrite}
under the same fit-extrapolation.
``RLCU'' uses the upper bound $\Gamma_{\mathrm{RLCU}}\le\exp(2\tilde t^{2}/r)$
at the optimal repetition $r^{*}=\tilde t/\sqrt{\gamma'}$.
(a) $M$ versus $n$ at four fixed accuracies $\epsilon\in\{10^{-1},10^{-2},10^{-3},10^{-4}\}$.
(b) $M$ versus $\epsilon$ at four fixed sizes $n\in\{10,20,30,40\}$.}
\label{fig:M_vs_n_eps}
\end{figure*}

\Cref{fig:M_vs_n_eps} compares the required number of circuit runs $M$ of the two algorithms for the XYZ chain of $n$ qubits, in an early-FTQC setting where the cost is the non-Clifford gate count and $\gamma'$ is the corresponding per-gate PEC-overhead rate.
The curves use $t=n$ and $k=2$ for Trotter, with $\gamma'=2\times10^{-7}$ chosen as an illustrative per-gate PEC-overhead rate for an early-FTQC regime, in which the logical error rate can be systematically reduced by increasing the code distance~\cite{fowler2012surface}. The same scaling behavior appears for any other $\gamma'\ll1$, so the comparison does not rely on this particular choice.
The XYZ chain has $L=3n$ Pauli terms, so $L$ is the relevant size variable and $\beta t\propto Lt$.

In~\cref{fig:M_vs_n} we plot $M$ versus $n$ at four fixed accuracies $\epsilon\in\{10^{-1},10^{-2},10^{-3},10^{-4}\}$, to be
compared with the noise-free baseline $M=\epsilon^{-2}$, the variance of an unbiased estimator of a bounded observable.
The relative performance of the two algorithms is governed by
$\log(\epsilon^{2}M)$, the
logarithm of the multiplicative PEC sampling overhead.
For RLCU this overhead is
$8\sqrt{\gamma'}\,\beta t+4\gamma_{c}'\beta t=O(\sqrt{\gamma}\,Lt)$ from
\cref{eq:epsilon_r_M_star}, independent of $\epsilon$.
For Trotter in the regime $\epsilon\ll\epsilon_{c}$ it is
$2k(\epsilon_{c}/\epsilon)^{1/k}$ from
\cref{eq:M_trotter}, which scales as $O(\gamma\,(Lt)^{(k+1)/k}\epsilon^{-1/k})$.
Because $\gamma\ll1$, the $\sqrt{\gamma}$ of the RLCU overhead exceeds
the higher power $\gamma$ of the Trotter overhead, so RLCU incurs the larger overhead at fixed $\epsilon$ in the moderate-accuracy regime.
As $\epsilon$ decreases from $10^{-1}$ to $10^{-4}$, the
$\epsilon$-dependent Trotter overhead increases toward the $\epsilon$-independent RLCU overhead,
which is unchanged apart from the $\epsilon^{-2}$ shot-noise term.
The commutator
prefactor $\alpha^{(2)}_{\rm comm}$ overstates the diamond-distance
$\alpha_{2}$ by roughly an order of magnitude, so the ``Trotter bound'' overestimates the Trotter cost relative to the tight ``Trotter'' estimate.

In~\cref{fig:M_vs_eps}, at large $\epsilon$ the Trotter cost approaches
$M\simeq\epsilon^{-2}$ and no longer depends on $n$. In contrast, the multiplicative overhead $\epsilon^{2}M$ of each of the four RLCU curves is independent of $\epsilon$, and consecutive curves differ only by the factor $\exp[8\sqrt{\gamma'}\,\beta t]$, which grows exponentially with $n$.
As $\epsilon$ drops below $\epsilon_{c}$, the Trotter cost enters the exponential
barrier $\exp[2k(\epsilon_{c}/\epsilon)^{1/k}]$ of \cref{eq:M_trotter}.
Because $\epsilon_{c}$ grows with $L$, this barrier begins at larger $\epsilon$ for larger $L$,
and at larger $\epsilon$ for the ``Trotter bound'' than for the tight ``Trotter'' estimate.
The two algorithms therefore cross.
Trotter is cheaper at moderate $\epsilon$, while RLCU is preferable in the high-precision
regime $\epsilon\ll\epsilon_{c}$, where the Trotter barrier exceeds the
RLCU overhead.

The crossover originates in the different dependence of the two algorithms on $\epsilon$.
The RLCU gate count is fixed by $\beta t$ and the repetition number $r$ and is independent of $\epsilon$, so a looser accuracy cannot reduce it, whereas the Trotter circuit depth $d=(\alpha_k/\epsilon)^{1/k}$ decreases with $\epsilon$ and lowers the gate count $dL$.
A loose target accuracy therefore reduces the Trotter cost but not the RLCU cost.
This behavior is specific to the geometrically local XYZ chain, where $\alpha^{(k)}_{\mathrm{comm}}=O(L)$. For Hamiltonians that are not geometrically local, for example the Hamiltonians with long-range interactions, this quantity grows more rapidly with the number of terms. The crossover point, and hence the regime in which each algorithm is preferable, then shifts away from that shown in~\cref{fig:M_vs_n_eps}.

\subsection{The cost of characterizing the noise model}
\label{sec:gst}
The analysis above assumes perfect knowledge of the noise model.
In practice, implementing PEC requires an explicit estimate of the noise parameters, which introduces an additional sampling cost. Here, we assume that we characterize the elementary gate operations, e.g., a universal gate set. For simplicity, we also assume that the elementary gate errors are static across the entire quantum circuit.
Let $\mathcal{N}_k$ denote a noise channel associated with the $k$th gate, 
and let $\mathcal{N}'^{-1}_k$ be its estimated inverse noise channel obtained from gate set tomography (GST)~\cite{nielsen2021gate,greenbaum2015introduction,endo2018practical}\@. Scalable Pauli-noise learning provides an alternative characterization route~\cite{flammia2020efficient,vandenberg2023probabilistic}.
Then, we write the maximum estimation error as
\begin{equation}
    \Delta \gamma := \max_k D(\mathcal{N}'^{-1}_k \circ \mathcal{N}_k, \mathcal{I}),
\end{equation}
and assume the standard scaling with respect to the number of GST samples $M_g$,
\begin{equation}
\label{eq:gst_scaling}
    \Delta \gamma
    =
    O\!\left(\frac{1}{\sqrt{M_g}}\right),
\end{equation}
up to constants that depend on the particular GST protocol and the underlying noise.

\paragraph*{Bias induced by characterization errors --}
When PEC is constructed using the estimated noise channels $\mathcal{N}'_k$ instead of the true $\mathcal{N}_k$,
the resulting estimator $\hat{\mu}$ for an observable expectation value generally becomes biased.
A simple bound is obtained by summing the worst-case contribution of each gate:
\begin{equation}
\label{eq:bias_bound_general}
    \bigl|\mathrm{Bias}(\hat{\mu})\bigr|
    \le O(N_g/\sqrt{M_g}),
\end{equation}
where $N_g$ is the number of elementary gates in the executed circuit (corresponding to $N_G$ in the general PEC formalism of \cref{sec:pec}).
See \cref{app:bias_bound} for the derivation.
For the RLCU algorithm, as we discussed in \cref{subsec:rlcu_baseline_noise}, the estimator is further rescaled by the sampling overhead $\Gamma_{\mathrm{RLCU}}$, and the corresponding bias bound becomes $\Gamma_{\mathrm{RLCU}}\,N_g\,\Delta\gamma$.
Using the GST scaling~\cref{eq:gst_scaling}, this translates into a sufficient number of GST samples
\begin{equation}
\label{eq:Mg_requirement_general}
    M_g
    =O
    \bigg(\frac{N_g^2}{\epsilon^2}\bigg),
\end{equation}
where $\epsilon$ is the target accuracy for the observable expectation value.
For the RLCU algorithm, the sampling cost scales as 
\begin{equation}
\label{eq:Mg_requirement_rlcu}
    M_g
    =O
    \bigg(\frac{\Gamma_{\rm RLCU}^2 N_g^2}{\epsilon^2}\bigg).
\end{equation}
To assess the relative weight of this characterization cost, 
we define the simulation cost $\mathcal{R} := d LM$ (for Suzuki--Trotter) or $\mathcal{R} := r M$ (for RLCU)
as the product of the per-circuit gate count (up to a constant factor) and the number of circuit runs $M$,
which serves as a proxy for the total runtime of the Hamiltonian simulation.
Then, the ratio $M_g / \mathcal{R}$ quantifies whether the characterization cost is 
a significant fraction of the total simulation budget.
We now evaluate this ratio for both the Suzuki--Trotter and RLCU algorithms.

\subsubsection{Suzuki--Trotter and GST Cost}
\label{subsec:st_baseline_gst}

We first benchmark the characterization cost against the conventional Suzuki--Trotter approach
in two regimes of the target accuracy in~\cref{eq:M_trotter}.
Fixing a target mean-squared error $\epsilon^2$, we can express the optimal circuit depth
$d^\star$ and the required number of circuit runs $M$ as functions of $\epsilon$.
In the regime $\epsilon \ll \epsilon_c$, from~\cref{eq:optimal_d} and~\cref{eq:M_trotter}, we have
\begin{align}
\label{eq:st_dstar}
    d^\star
    &\simeq
    \frac{k}{\gamma' L}\left(\frac{\epsilon_c}{\epsilon}\right)^{1/k},
    \\
\label{eq:st_M}
    M
    &\simeq
    \frac{1}{\epsilon^2}\left(\frac{\epsilon_c}{\epsilon}\right)^{1/k}
    \exp\!\left[
        2k\left(\frac{\epsilon_c}{\epsilon}\right)^{1/k}
    \right],
\end{align}
and in the regime $\epsilon \gg \epsilon_c$, we have
\begin{align}
\label{eq:dstar_large_eps}
    d^\star
    &\simeq
    \frac{k}{\gamma' L}\left(\frac{\epsilon_c}{\epsilon}\right)^{2/(2k+1)},
    \\
\label{eq:st_M_large_eps}
    M
    &=
    \frac{1}{\epsilon^2}\left(1 + 2k\left(\frac{\epsilon_c}{\epsilon}\right)^{2/(2k+1)}\right).
\end{align}

By noting that the number of elementary gates is $N_g = d^\star L$,
we obtain the GST sampling cost to achieve the target accuracy $\epsilon$ as

\begin{equation}
\label{eq:Ng_epsilon}
    M_g =
    \begin{cases}
        O\left(\left[\dfrac{k}{\gamma' \epsilon}\left(\dfrac{\epsilon_c}{\epsilon}\right)^{1/k}\right]^2 \right)
        & \epsilon \ll \epsilon_c, \\[8pt]
       O\left( \left(\dfrac{k}{\gamma' \epsilon_c}\right)^2  \right)
        & \epsilon = \epsilon_c, \\[8pt]
       O\left( \left[\dfrac{k}{\gamma' \epsilon}\left(\dfrac{\epsilon_c}{\epsilon}\right)^{2/(2k+1)}\right]^2 \right)
        & \epsilon \gg \epsilon_c,
    \end{cases}
\end{equation}
and by using $\mathcal{R} = d^\star L M$, the ratio between the GST sampling cost and the simulation cost scales as
\begin{equation}
\label{eq:ratio_Mg_R_st}
    \frac{M_g}{\mathcal{R}} =
    \begin{cases}
        O\left(\dfrac{k}{\gamma'}\exp\!\left[-2k\left(\dfrac{\epsilon_c}{\epsilon}\right)^{1/k}\right]\right)
        & \epsilon \ll \epsilon_c, \\[8pt]
        O\left(\dfrac{k}{\gamma'}\,e^{-2k} \right)
        & \epsilon = \epsilon_c, \\[8pt]
        O\left(\dfrac{k}{\gamma'}\left(\dfrac{\epsilon_c}{\epsilon}\right)^{2/(2k+1)}  \right)
        & \epsilon \gg \epsilon_c.
    \end{cases}
\end{equation}
In the regime $\epsilon \ll \epsilon_c$, the exponential factor suppresses the ratio,
so that the GST cost is negligible compared to the Hamiltonian-simulation cost because the PEC sampling overhead dominates the total budget.
In the regime $\epsilon \gg \epsilon_c$, the ratio scales as $O(k/\gamma')$
up to a polynomial correction in $\epsilon_c/\epsilon$.
This indicates that the characterization cost becomes
 more significant as the per-gate error $\gamma$ decreases.



\subsubsection{Randomized LCU}
For the RLCU algorithm, following a similar argument as in \cref{eq:bias_rlcu_bound},
the expected gate count per circuit is $N_g = 2\left(r + (\beta t)^2/r\right)$, with the factor $2$ from the two control unitaries per segment.
As discussed in~\cref{eq:epsilon_r_M_approx}, the choice of the repetition number depends on
whether $\sqrt{\gamma'}\,\beta t$ exceeds $1$.
In the regime $\sqrt{\gamma'}\,\beta t \ge 1$, we use the optimized value $r^* = \beta t/\sqrt{\gamma'}$, which offers a smaller sampling cost and shallower circuit depth than the choice of $r = (\beta t)^2$, while in the regime $\sqrt{\gamma'}\,\beta t \le 1$, we set $r = (\beta t)^2$.
Indeed, since the circuit depth is proportional to $N_g \simeq 2r$, the ratio of the two repetition numbers is $r^*/r = 1/(\sqrt{\gamma'}\,\beta t)$, so $r^*$ yields the shallower circuit once $\sqrt{\gamma'}\,\beta t \ge 1$. A shallower circuit is advantageous for the GST characterization because the sampling cost scales as $M_g = O(N_g^2/\epsilon^2)$ (\cref{eq:Mg_requirement_general}), i.e., quadratically in the gate count, so reducing the depth directly lowers the characterization cost.
In both cases, the second term is negligible compared to the first, so we 
approximate $N_g \simeq 2r$.
Using~\cref{eq:Mg_requirement_rlcu}, 
the GST sampling cost to
achieve the target accuracy $\epsilon$ is
\begin{equation}
\label{eq:Mg_rlcu}
    M_g =
    \begin{cases}
        O\left(\Gamma_{\mathrm{RLCU}}^{2}\,\dfrac{(\beta t)^2}{\gamma'\,\epsilon^2} \right)
        & \sqrt{\gamma'}\,\beta t \ge 1, \\[8pt]
        O\left(\Gamma_{\mathrm{RLCU}}^{2}\,\dfrac{(\beta t)^4}{\epsilon^2} \right)
        & \sqrt{\gamma'}\,\beta t \le 1.
    \end{cases}
\end{equation}

As in the Suzuki--Trotter case, we compare the characterization cost $M_g$ with the overall
simulation cost $\mathcal{R} = r^\star M$.
Then, the ratio between the GST sampling cost and the simulation cost scales as
\begin{equation}
\label{eq:ratio_Mg_R_rlcu}
    \frac{M_g}{\mathcal{R}}
    =
    \begin{cases}
        O\left(\dfrac{\beta t}{\sqrt{\gamma'}} \right)\,
        & \sqrt{\gamma'}\,\beta t \ge 1, \\[8pt]
        O\left((\beta t)^2 \right)
        & \sqrt{\gamma'}\,\beta t \le 1.
    \end{cases}
\end{equation}
Note that both the sampling cost $M$ and the characterization cost $M_g$ are proportional to $\Gamma_{\mathrm{RLCU}}^2$, which cancels out in the ratio $M_g/\mathcal{R}$.

Importantly, in the Suzuki--Trotter case, and in the RLCU case until the ratio saturates at $O((\beta t)^2)$ for $\sqrt{\gamma'}\,\beta t \le 1$,
the resource ratio $M_g/\mathcal{R}$ increases as $\gamma'$ decreases
(i.e., as gate fidelity improves).
Suppose the per-gate error rate is $\gamma$. To remove the noise-induced bias via PEC, the noise model must be determined to a precision higher than $\gamma$ itself. 
Otherwise, the residual characterization error would dominate the error-mitigated result.
Consequently, from \cref{eq:gst_scaling}, the characterization cost scales as $M_g = O(\gamma^{-2})$, while the simulation cost $\mathcal{R}$ grows only as $O(\gamma^{-1})$ through the optimal depth, so the ratio $M_g/\mathcal{R}$ grows in the high-fidelity regime.

\section{Error-mitigated Hamiltonian simulation with the space-time noise inversion method}
\label{sec:sni}
In the previous section, the characterization cost of the standard PEC via the standard GST can become significant, particularly as the per-gate error decreases (\cref{eq:ratio_Mg_R_st}).
To improve this scaling, we now consider the space-time noise inversion (SNI) method~\cite{xie2026noise}, a PEC variant that characterizes only the aggregate error probability of the entire circuit rather than individual gate errors.
This method targets the early-FTQC setting, in which logical error rates are low but the number of logical qubits is limited; combining QEM with quantum error correction in this setting has been studied in Refs.~\cite{suzuki2022quantum,piveteau2021error}.

\subsection{Space-time noise inversion}
The space-time noise inversion (SNI) method aims to invert the noise across the entire quantum circuit by treating it as a single space-time noise model. More specifically, when the noisy quantum process in the quantum circuit is described as $\prod_{k=1}^{N_G} \mathcal{N}_k \mathcal{U}_k $, the space-time noise is described as $\mathcal{N}_{\rm ST}= \mathcal{N}_{N_G} \otimes \mathcal{N}_{N_G-1} \otimes \cdots \otimes \mathcal{N}_1$. This method assumes that the noise model is described by stochastic Pauli noise via twirling techniques~\cite{knill2005quantum,wallman2016noise}. Then, the space-time noise is represented as $\mathcal{N}_{\rm ST}= (1-p_{\rm ST}) \mathcal{I} + p_{\rm ST} \mathcal{E}$, where $\mathcal{E}$ corresponds to the process for Pauli-error events and $p_{\rm ST}$ is the total error probability of the circuit. Then, SNI inverts the entire space-time noise as
\begin{equation}
\mathcal{N}_{\rm ST}^{-1} = \sum_{l=0}^{\infty} \frac{(-1)^l p_{\rm ST}^l}{(1-p_{\rm ST})^{l+1}} \mathcal{E}^l
\label{eq:sni}
\end{equation}
for $p_{\rm ST} < 1/2$. The sampling overhead for this quasi-probability decomposition is $\Gamma= \sum_{l=0}^\infty p_{\rm ST}^l/(1-p_{\rm ST})^{l+1}=1/(1-2 p_{\rm ST})$. To perform~\cref{eq:sni}, (i) we only need to know the single error probability $p_{\rm ST}$ and (ii) sample $\mathcal{E}^l$. To do so, we use a quantum circuit that performs Bell measurements for each noisy quantum operation $\mathcal{N}_k \mathcal{U}_k$ using twice the number of logical qubits. Crucially, while the sampling is performed at the gadget level, these samples are aggregated to characterize the single global error probability $p_{\rm ST}$. The sampling cost for characterizing the error probability $p_{\rm ST}$ scales as $M_{p_{\rm ST}}=O (\epsilon^{-2} (1- 2p_{\rm ST})^{-4})$, with the required number of circuit runs for the computation circuit being $M = O(\epsilon^{-2} (1-2p_{\rm ST})^{-2})$ to certify the total computation accuracy of the error-mitigated computation to $\epsilon$. Note that the SNI method is also compatible with the randomized compiling~\cite{wallman2016noise}, mid-circuit measurement, and feedback operations; therefore, SNI is compatible with the RLCU algorithm as well. While the SNI sampling cost is divergent at $p_{\rm ST}=1/2$, this problem can be circumvented by separating the quantum circuit into $s$ segments. By applying SNI to each segment individually where the local error rate is well below $1/2$, the total sampling cost is determined by the product of the costs for each segment, thus avoiding the singularity.

\subsection{Improved characterization scaling via space-time noise inversion}
We now consider applying PEC via the space-time noise inversion method to the Hamiltonian simulation algorithms discussed above.
We assume that the total error probability of the entire quantum circuit is given by $p_{\rm ST}$.
Using the worst-case per-gate error rate $\gamma$ introduced in \cref{eq:gamma}, the 
total space-time error probability for 
a circuit consisting of $dL$-gates is bounded by
$
    p_{\rm ST}
    = 1 - (1-\gamma)^{dL}
    \simeq 1 - e^{-\gamma dL}.
$
If $p_{\rm ST} < 1/2$, the SNI method can be applied directly to the entire circuit.
However, when $p_{\rm ST} \ge 1/2$, the SNI sampling overhead diverges, and it becomes necessary to partition the circuit into $s$ segments, as discussed in the previous subsection.
Because both the Suzuki--Trotter formula and the RLCU algorithm consist of repeated applications of an identical circuit structure,
we assume that the total error probability is uniform across all segments.
Let $q = 1-(1-\gamma)^{L}$ denote the upper bound of the per-layer error probability. Note also that $q \le \gamma L$.
Under this assumption, let $q_{\rm ST}$ denote the space-time error probability per segment, given by
\begin{equation}
    \label{eq:qst_seg_bound}
    q_{\rm ST}
    = 1 - (1 - q)^{d/s}.
\end{equation}
In this segmented implementation, the sampling cost required to characterize the error probability $q_{\rm ST}$ scales as
\begin{equation}
    M_{q_{\rm ST}}
    = O\!\left(
        \epsilon^{-2}
        s^2
        \left(1 - 2 q_{\rm ST}\right)^{-4}
    \right),
\end{equation}
while the number of circuit executions required for the error-mitigated computation scales as
\begin{equation}
    \label{eq:M_sni_segment}
    M
    = O\!\left(
        \epsilon^{-2}
        \left(1 - 2 q_{\rm ST}\right)^{-2s}
    \right).
\end{equation}

\Cref{eq:M_sni_segment} may appear to differ from the PEC overhead scaling in 
\cref{eq:PEC-overhead}. 
However, the following inequality shows that it can be written in a closely related form:
\begin{equation}
    \exp\!\left( 4 q d \right)
    < (1-2q_{\rm ST})^{-2s}
    < \exp\!\left( \frac{4 q d}{1 - 2q d /s} \right).
\end{equation}

We detail the derivation of
this bound in~\cref{app:sni_segment_scaling}.
If we choose the number of segments $s$ so that $1 - 2qd/s = c$ for a positive constant $c$,
the upper bound becomes $M = O\!\left(\epsilon^{-2} \exp(4qd / c)\right)$.
Because this has the same functional form as \cref{eq:epsilon_d_m},
the optimization of the circuit depth $d$ follows the same argument that leads to \cref{eq:M_trotter},
with $L\gamma'$ replaced by $2q/c \le 2L\gamma/c$.
Furthermore, if we choose $s = 4 q d$, corresponding to $c = 1/2$, the scalings are, up to constant factors,
\begin{equation}
\begin{aligned}
    M_{q_{\rm ST}} &= O\!\left(\epsilon^{-2} {(Ld\gamma)}^2\right),\\
    M &= O\!\left(\epsilon^{-2} \exp(8 Ld\gamma)\right).
\end{aligned}
\end{equation}
Here we used $q_{\rm ST} \leq qd/s$ and $q \leq \gamma L$.

In the regime where $Ld\gamma = O(1)$, 
the characterization cost no longer increases as $\gamma\to 0$
while keeping $Ld\gamma = O(1)$,
unlike the per-gate GST approach, where $M_g$ scales as $O(\gamma^{-2})$ as $\gamma\to 0$ (\cref{eq:gst_scaling}).

To quantify the resource ratio between the characterization and simulation costs,
we evaluate the ratio $M_{q_{\rm ST}}/\mathcal{R}$
with $\mathcal{R}=d\,L\,M$.
Substituting $M_{q_{\rm ST}}=O(\epsilon^{-2}(qd)^{2})$ and
$M=\Omega(\epsilon^{-2}\exp(4qd))$ yields
\begin{equation}
    \frac{M_{q_{\rm ST}}}{\mathcal{R}}
    =O\!\left(\frac{q^{2}\,d}{L\,e^{4qd}}\right)
    =O\!\left(\gamma\cdot qd\,e^{-4qd}\right),
\end{equation}
where we used $q^{2}/L \le \gamma q$ from $q \le \gamma L$.
Because $x e^{-4x}\le 1/(4e)$ for all $x \ge 0$,
the noise-characterization cost is exponentially smaller than
the simulation cost whenever $qd \gg 1$.
Even in the worst case, $M_{q_{\rm ST}}/\mathcal{R}=O(1)$.

SNI avoids the divergence of the sampling cost for a small error rate observed in GST by characterizing only the aggregate
error probability $q_{\rm ST}$ per segment rather than the full
per-gate noise model, resulting in a characterization cost that is 
smaller or comparable to
the Hamiltonian simulation cost.
Note that the same analysis applies to the RLCU algorithm by replacing $dL$ with $N_g \simeq 2r$ (see the discussion preceding \cref{eq:Mg_rlcu}).
Specifically, the characterization and simulation costs become
$M_{q_{\rm ST}} = O(\epsilon^{-2}\Gamma_{\rm RLCU}^2 (r \gamma)^2)$ and
$M = O(\epsilon^{-2}\Gamma_{\rm RLCU}^2 \exp(8r\gamma))$, respectively.
With $\mathcal{R} = rM$, the ratio satisfies
$M_{q_{\rm ST}}/\mathcal{R} = O(1)$ by the same argument as in the Trotter case.
Because the segmented SNI overhead retains the same functional form $\exp(O(\gamma N_g))$ regardless of the underlying algorithm,
SNI can be used to efficiently characterize the noise model for Hamiltonian simulation algorithms that share this repeated layer structure.

\section{Conclusions and Discussions}
In this work, we propose optimizing the circuit depth in Hamiltonian simulation algorithms,
i.e., Trotter- and RLCU-based algorithms, to minimize the mean-squared error (MSE)
that accounts for both physical and algorithmic contributions.
For the Trotter-based simulation, we show that the sampling cost exhibits two distinct regimes
as a function of the target accuracy: a polynomial regime and an exponential regime.
This implies the existence of a critical error $\epsilon_c$, beyond which further accuracy
improvements require an exponentially growing number of circuit runs.
For the RLCU-based algorithm, we can realize an unbiased estimator of the
target observable via the RLCU construction.
Moreover, the circuit depth is a random variable with a bounded expectation, so that the end-to-end
performance is governed primarily by the sampling overhead rather than by an algorithmic bias.

Building on this, we optimize the repetition number $r$.
By balancing the terms $4\tilde t^{2}/r$ and $4\gamma' r$ in the exponent of the MSE bound,
we obtain an optimal choice $r^*=\tilde t/\sqrt{\gamma'}$, which yields a square-root improvement
in the dominant time dependence of the exponential sampling cost, from $O(\gamma (\beta t)^{2})$ (e.g., for $r=\tilde t^{2}$) to $O(\sqrt{\gamma}\,\beta t)$.
In addition, we quantify the characterization cost associated with gate set tomography (GST)
required for error mitigation, and we further evaluate how space-time noise inversion (SNI) can
improve the scaling of this cost. 

We also remark on the choice of the cost measure. In this work, we optimize the number of circuit runs $M$ as a function of the circuit depth and the repetition number. Depending on the physical platform of the quantum computer and on the underlying error-correcting code, other cost measures, such as the total execution time or the total gate count, may be more appropriate. The same optimization framework applies to such measures as well.

Note that our approach can be naturally extended to other types of Hamiltonian simulation algorithms. In particular, the recently proposed TE-PAI method~\cite{kiumi2025te} enables the computation of unbiased estimators of expectation values via random sampling of a quantum circuit, based on the quasi-probability decomposition of the Trotter algorithm with the probabilistic angle interpolation~\cite{koczor2024probabilistic}. Because the sample complexity of this method can be controlled at the cost of the depth of the circuit (denoted as the Q parameter in Ref.~\cite{kiumi2025te}), similarly to the parameter $r$ in RLCU, the depth can likewise be optimized when combined with QEM. 

Furthermore, while we focus on the PEC method, it is worth investigating other QEM methods for Hamiltonian simulation. For example, the virtual distillation~\cite{huggins2021virtual,koczor2021exponential} and the extrapolation~\cite{temme2017error,mohammadipour2025direct} do not require the explicit information of noise, so they may reduce the total complexity of error-mitigated Hamiltonian simulation. In addition, the unification of QEM methods, via e.g., generalized quantum subspace expansion~\cite{yoshioka2022generalized,yang2025resource}, may further reduce both the physical noise and the algorithmic error. 


Finally, the optimization of algorithmic resources for minimizing QEM sampling overhead is likely relevant beyond real-time Hamiltonian simulation. Similar optimization principles may apply to a broad range of quantum algorithms, such as quantum linear system solvers~\cite{harrow2009quantum,childs2017quantum,wang2024qubit}, ground-state estimation~\cite{ge2019faster,zeng2021universal}, and density-matrix exponentiation~\cite{lloyd2014quantum,wada2025state}. In addition, recent studies of quantum-classical hybrid implementations of linear-combination-of-unitaries methods indicate that quantum resources, including circuit depth and ancilla usage, can be traded against classical sampling overhead~\cite{wada2025tradeoffs}.
It would therefore be important to extend such resource-trade-off analyses by explicitly incorporating the cost of quantum error mitigation and identifying the optimal quantum-classical balance under noise.

\section*{Acknowledgments}
This work was supported by JSPS KAKENHI Grant Numbers 24KJ0892, 20H01824 and 24K00543, 
the Center of Innovation for Sustainable Quantum AI (SQAI), 
and JST Grant Number JPMJPF2221;
JST [Moonshot R\&D] Grant No.~JPMJMS2061; 
MEXT Q-LEAP, Grant No.~JPMXS0120319794 and No.~JPMXS0118067285; 
and JST CREST Grant No.~JPMJCR23I4 and No.~JPMJCR25I4.

\section*{Author contributions}
K.M. developed the theoretical analysis and performed the calculations.
S.E. contributed to the theoretical calculations.
S.T. and S.E. supervised the project.
All authors contributed to writing the manuscript.
A large language model (Claude, Anthropic) was used for grammar checking and editing of the text, and for assistance with parts of the proofs in Appendices~\ref{app:proof_avg_pec_overhead} and~\ref{app:bias_bound}.

\bibliographystyle{apsrev4-1}
\bibliography{bib}

@article{fowler2012surface,
  title = {Surface codes: Towards practical large-scale quantum computation},
  author = {Fowler, Austin G. and Mariantoni, Matteo and Martinis, John M. and Cleland, Andrew N.},
  journal = {Phys. Rev. A},
  volume = {86},
  issue = {3},
  pages = {032324},
  numpages = {48},
  year = {2012},
  publisher = {American Physical Society},
  doi = {10.1103/PhysRevA.86.032324}
}

@article{lloyd1996universal,
  doi={10.1126/science.273.5278.1073},
  title={Universal quantum simulators},
  author={Lloyd, Seth},
  journal={Science},
  volume={273},
  number={5278},
  pages={1073--1078},
  year={1996},
  publisher={American Association for the Advancement of Science}
}

@article{childs2021,
  title = {Theory of Trotter Error with Commutator Scaling},
  author = {Childs, Andrew M. and Su, Yuan and Tran, Minh C. and Wiebe, Nathan and Zhu, Shuchen},
  journal = {Phys. Rev. X},
  volume = {11},
  issue = {1},
  pages = {011020},
  numpages = {49},
  year = {2021},
  month = {Feb},
  publisher = {American Physical Society},
  doi = {10.1103/PhysRevX.11.011020},
  url = {https://link.aps.org/doi/10.1103/PhysRevX.11.011020}
}

@article{Chakraborty2024,
  doi = {10.22331/q-2024-10-10-1496},
  url = {https://doi.org/10.22331/q-2024-10-10-1496},
  title = {Implementing any {L}inear {C}ombination of {U}nitaries on {I}ntermediate-term {Q}uantum {C}omputers},
  author = {Chakraborty, Shantanav},
  journal = {{Quantum}},
  issn = {2521-327X},
  publisher = {{Verein zur F{\"{o}}rderung des Open Access Publizierens in den Quantenwissenschaften}},
  volume = {8},
  pages = {1496},
  month = oct,
  year = {2024}
}

@article{berry2015simulating,
  doi={10.1103/PhysRevLett.114.090502},
  title={Simulating Hamiltonian dynamics with a truncated Taylor series},
  author={Berry, Dominic W and Childs, Andrew M and Cleve, Richard and Kothari, Robin and Somma, Rolando D},
  journal={Physical review letters},
  volume={114},
  number={9},
  pages={090502},
  year={2015},
  publisher={APS}
}

@article{low2019hamiltonian,
  doi={10.22331/q-2019-07-12-163},
  title={Hamiltonian simulation by qubitization},
  author={Low, Guang Hao and Chuang, Isaac L},
  journal={Quantum},
  volume={3},
  pages={163},
  year={2019},
  publisher={Verein zur F{\"o}rderung des Open Access Publizierens in den Quantenwissenschaften}
}

@misc{kitaev1995quantum,
  title={Quantum measurements and the Abelian stabilizer problem},
  author={Kitaev, A Yu},
  eprint={quant-ph/9511026},
  archivePrefix={arXiv},
  primaryClass={quant-ph},
  doi={10.48550/arXiv.quant-ph/9511026},
  year={1995}
}

@article{aspuru2005simulated,
  doi={10.1126/science.1113479},
  title={Simulated quantum computation of molecular energies},
  author={Aspuru-Guzik, Al{\'a}n and Dutoi, Anthony D and Love, Peter J and Head-Gordon, Martin},
  journal={Science},
  volume={309},
  number={5741},
  pages={1704--1707},
  year={2005},
  publisher={American Association for the Advancement of Science}
}

@article{lin2022heisenberg,
  doi={10.1103/PRXQuantum.3.010318},
  title={Heisenberg-limited ground-state energy estimation for early fault-tolerant quantum computers},
  author={Lin, Lin and Tong, Yu},
  journal={PRX quantum},
  volume={3},
  number={1},
  pages={010318},
  year={2022},
  publisher={APS}
}

@article{wan2022randomized,
  doi={10.1103/PhysRevLett.129.030503},
  title={Randomized quantum algorithm for statistical phase estimation},
  author={Wan, Kianna and Berta, Mario and Campbell, Earl T},
  journal={Physical Review Letters},
  volume={129},
  number={3},
  pages={030503},
  year={2022},
  publisher={APS}
}

@article{an2023linear,
  doi={10.1103/PhysRevLett.131.150603},
  title={Linear combination of Hamiltonian simulation for nonunitary dynamics with optimal state preparation cost},
  author={An, Dong and Liu, Jin-Peng and Lin, Lin},
  journal={Physical Review Letters},
  volume={131},
  number={15},
  pages={150603},
  year={2023},
  publisher={APS}
}

@article{an2026quantum,
  doi={10.1007/s00220-025-05509-w},
  title={Quantum Algorithm for Linear Non-unitary Dynamics with Near-Optimal Dependence on All Parameters},
  author={An, Dong and Childs, Andrew M and Lin, Lin},
  journal={Communications in Mathematical Physics},
  volume={407},
  number={1},
  pages={19},
  year={2026},
  publisher={Springer}
}

@article{childs2017quantum,
  doi={10.1137/16M1087072},
  title={Quantum algorithm for systems of linear equations with exponentially improved dependence on precision},
  author={Childs, Andrew M and Kothari, Robin and Somma, Rolando D},
  journal={SIAM Journal on Computing},
  volume={46},
  number={6},
  pages={1920--1950},
  year={2017},
  publisher={SIAM}
}

@article{wang2024qubit,
  doi={10.1103/PRXQuantum.5.020324},
  title={Qubit-efficient randomized quantum algorithms for linear algebra},
  author={Wang, Samson and McArdle, Sam and Berta, Mario},
  journal={PRX quantum},
  volume={5},
  number={2},
  pages={020324},
  year={2024},
  publisher={APS}
}

@article{knee2015optimal,
  doi={10.1103/PhysRevA.91.052327},
  title={Optimal Trotterization in universal quantum simulators under faulty control},
  author={Knee, George C and Munro, William J},
  journal={Physical Review A},
  volume={91},
  number={5},
  pages={052327},
  year={2015},
  publisher={APS}
}

@article{endo2019mitigating,
  doi={10.1103/PhysRevA.99.012334},
  title={Mitigating algorithmic errors in a Hamiltonian simulation},
  author={Endo, Suguru and Zhao, Qi and Li, Ying and Benjamin, Simon and Yuan, Xiao},
  journal={Physical Review A},
  volume={99},
  number={1},
  pages={012334},
  year={2019},
  publisher={APS}
}

@misc{hakkaku2025data,
  title={Data-Efficient Error Mitigation for Physical and Algorithmic Errors in a Hamiltonian Simulation},
  author={Hakkaku, Shigeo and Suzuki, Yasunari and Tokunaga, Yuuki and Endo, Suguru},
  eprint={2503.05052},
  archivePrefix={arXiv},
  primaryClass={quant-ph},
  doi={10.48550/arXiv.2503.05052},
  year={2025}
}

@article{mohammadipour2025direct,
  doi={10.22331/q-2025-11-14-1909},
  title={Direct Analysis of Zero-Noise Extrapolation: Polynomial Methods, Error Bounds, and Simultaneous Physical-Algorithmic Error Mitigation},
  author={Mohammadipour, Pegah and Li, Xiantao},
  journal={Quantum},
  volume={9},
  pages={1909},
  year={2025},
  publisher={Verein zur F{\"o}rderung des Open Access Publizierens in den Quantenwissenschaften}
}

@misc{xu2025exponentially,
  title={Exponentially Decaying Quantum Simulation Error with Noisy Devices},
  author={Xu, Jue and Zhao, Chu and Fan, Junyu and Zhao, Qi},
  eprint={2504.10247},
  archivePrefix={arXiv},
  primaryClass={quant-ph},
  doi={10.48550/arXiv.2504.10247},
  year={2025}
}

@article{cai2023quantum,
  doi={10.1103/RevModPhys.95.045005},
  title={Quantum error mitigation},
  author={Cai, Zhenyu and Babbush, Ryan and Benjamin, Simon C and Endo, Suguru and Huggins, William J and Li, Ying and McClean, Jarrod R and O’Brien, Thomas E},
  journal={Reviews of Modern Physics},
  volume={95},
  number={4},
  pages={045005},
  year={2023},
  publisher={APS}
}

@article{endo2021hybrid,
  doi={10.7566/JPSJ.90.032001},
  title={Hybrid quantum-classical algorithms and quantum error mitigation},
  author={Endo, Suguru and Cai, Zhenyu and Benjamin, Simon C and Yuan, Xiao},
  journal={Journal of the Physical Society of Japan},
  volume={90},
  number={3},
  pages={032001},
  year={2021},
  publisher={The Physical Society of Japan}
}

@article{temme2017error,
  doi={10.1103/PhysRevLett.119.180509},
  title={Error mitigation for short-depth quantum circuits},
  author={Temme, Kristan and Bravyi, Sergey and Gambetta, Jay M},
  journal={Physical review letters},
  volume={119},
  number={18},
  pages={180509},
  year={2017},
  publisher={APS}
}

@article{li2017efficient,
  doi={10.1103/PhysRevX.7.021050},
  title={Efficient variational quantum simulator incorporating active error minimization},
  author={Li, Ying and Benjamin, Simon C},
  journal={Physical Review X},
  volume={7},
  number={2},
  pages={021050},
  year={2017},
  publisher={APS}
}

@article{endo2018practical,
  doi={10.1103/PhysRevX.8.031027},
  title={Practical quantum error mitigation for near-future applications},
  author={Endo, Suguru and Benjamin, Simon C and Li, Ying},
  journal={Physical Review X},
  volume={8},
  number={3},
  pages={031027},
  year={2018},
  publisher={APS}
}

@misc{cai2021practical,
  title={A practical framework for quantum error mitigation},
  author={Cai, Zhenyu},
  eprint={2110.05389},
  archivePrefix={arXiv},
  primaryClass={quant-ph},
  doi={10.48550/arXiv.2110.05389},
  year={2021}
}

@article{berry2007efficient,
  doi={10.1007/s00220-006-0150-x},
  title={Efficient quantum algorithms for simulating sparse Hamiltonians},
  author={Berry, Dominic W and Ahokas, Graeme and Cleve, Richard and Sanders, Barry C},
  journal={Communications in Mathematical Physics},
  volume={270},
  number={2},
  pages={359--371},
  year={2007},
  publisher={Springer}
}

@article{xie2026noise,
  doi={10.1103/2n23-4qg8},
  title={Noise-Agnostic Unbiased Quantum Error Mitigation for Logical Qubits},
  author={Xie, Haipeng and Yoshioka, Nobuyuki and Tsubouchi, Kento and Li, Ying},
  journal={Physical Review Letters},
  volume={136},
  number={1},
  pages={010603},
  year={2026},
  publisher={APS}
}

@article{kiumi2025te,
  doi={10.1088/2058-9565/ae1160},
  title={TE-PAI: exact time evolution by sampling random circuits},
  author={Kiumi, Chusei and Koczor, B{\'a}lint},
  journal={Quantum Science and Technology},
  volume={10},
  number={4},
  pages={045071},
  year={2025},
  publisher={IOP Publishing}
}

@article{koczor2024probabilistic,
  doi={10.1103/PhysRevLett.132.130602},
  title={Probabilistic interpolation of quantum rotation angles},
  author={Koczor, B{\'a}lint and Morton, John JL and Benjamin, Simon C},
  journal={Physical Review Letters},
  volume={132},
  number={13},
  pages={130602},
  year={2024},
  publisher={APS}
}

@article{huggins2021virtual,
  doi={10.1103/PhysRevX.11.041036},
  title={Virtual distillation for quantum error mitigation},
  author={Huggins, William J and McArdle, Sam and O’Brien, Thomas E and Lee, Joonho and Rubin, Nicholas C and Boixo, Sergio and Whaley, K Birgitta and Babbush, Ryan and McClean, Jarrod R},
  journal={Physical Review X},
  volume={11},
  number={4},
  pages={041036},
  year={2021},
  publisher={APS}
}

@article{koczor2021exponential,
  doi={10.1103/PhysRevX.11.031057},
  title={Exponential error suppression for near-term quantum devices},
  author={Koczor, B{\'a}lint},
  journal={Physical Review X},
  volume={11},
  number={3},
  pages={031057},
  year={2021},
  publisher={APS}
}

@article{yoshioka2022generalized,
  doi={10.1103/PhysRevLett.129.020502},
  title={Generalized quantum subspace expansion},
  author={Yoshioka, Nobuyuki and Hakoshima, Hideaki and Matsuzaki, Yuichiro and Tokunaga, Yuuki and Suzuki, Yasunari and Endo, Suguru},
  journal={Physical Review Letters},
  volume={129},
  number={2},
  pages={020502},
  year={2022},
  publisher={APS}
}

@article{yang2025resource,
  doi={10.1103/PhysRevApplied.23.054021},
  title={Resource-efficient generalized quantum subspace expansion},
  author={Yang, Bo and Yoshioka, Nobuyuki and Harada, Hiroyuki and Hakkaku, Shigeo and Tokunaga, Yuuki and Hakoshima, Hideaki and Yamamoto, Kaoru and Endo, Suguru},
  journal={Physical Review Applied},
  volume={23},
  number={5},
  pages={054021},
  year={2025},
  publisher={APS}
}

@article{watson2025exponentially,
  doi={10.1103/kw39-yxq5},
  title={Exponentially reduced circuit depths using {T}rotter error mitigation},
  author={Watson, James D and Watkins, Jacob},
  journal={PRX Quantum},
  volume={6},
  number={3},
  pages={030325},
  year={2025},
  publisher={APS}
}

@article{PhysRevLett.123.070503,
  title = {Random Compiler for Fast Hamiltonian Simulation},
  author = {Campbell, Earl},
  journal = {Phys. Rev. Lett.},
  volume = {123},
  issue = {7},
  pages = {070503},
  numpages = {5},
  year = {2019},
  month = {Aug},
  publisher = {American Physical Society},
  doi = {10.1103/PhysRevLett.123.070503},
  url = {https://link.aps.org/doi/10.1103/PhysRevLett.123.070503}
}

@article{regula2021operational,
  doi={10.22331/q-2021-08-09-522},
  title={Operational applications of the diamond norm and related measures in quantifying the non-physicality of quantum maps},
  author={Regula, Bartosz and Takagi, Ryuji and Gu, Mile},
  journal={Quantum},
  volume={5},
  pages={522},
  year={2021},
  publisher={Verein zur F{\"o}rderung des Open Access Publizierens in den Quantenwissenschaften}
}

@article{jiang2021physical,
  doi={10.22331/q-2021-12-07-600},
  title={Physical implementability of linear maps and its application in error mitigation},
  author={Jiang, Jiaqing and Wang, Kun and Wang, Xin},
  journal={Quantum},
  volume={5},
  pages={600},
  year={2021},
  publisher={Verein zur F{\"o}rderung des Open Access Publizierens in den Quantenwissenschaften}
}

@article{harrow2009quantum,
  doi={10.1103/PhysRevLett.103.150502},
  title={Quantum algorithm for linear systems of equations},
  author={Harrow, Aram W and Hassidim, Avinatan and Lloyd, Seth},
  journal={Physical review letters},
  volume={103},
  number={15},
  pages={150502},
  year={2009},
  publisher={APS}
}

@article{ge2019faster,
  doi={10.1063/1.5027484},
  title={Faster ground state preparation and high-precision ground energy estimation with fewer qubits},
  author={Ge, Yimin and Tura, Jordi and Cirac, J Ignacio},
  journal={Journal of Mathematical Physics},
  volume={60},
  number={2},
  year={2019},
  publisher={AIP Publishing}
}

@misc{zeng2021universal,
  title={Universal quantum algorithmic cooling on a quantum computer},
  author={Zeng, Pei and Sun, Jinzhao and Yuan, Xiao},
  eprint={2109.15304},
  archivePrefix={arXiv},
  primaryClass={quant-ph},
  doi={10.48550/arXiv.2109.15304},
  year={2021}
}

@article{lloyd2014quantum,
  doi={10.1038/nphys3029},
  title={Quantum principal component analysis},
  author={Lloyd, Seth and Mohseni, Masoud and Rebentrost, Patrick},
  journal={Nature physics},
  volume={10},
  number={9},
  pages={631--633},
  year={2014},
  publisher={Nature Publishing Group UK London}
}

@misc{wada2025state,
  title={State-to-Hamiltonian conversion with a few copies},
  author={Wada, Kaito and Kato, Jumpei and Harada, Hiroyuki and Yamamoto, Naoki},
  eprint={2509.14791},
  archivePrefix={arXiv},
  primaryClass={quant-ph},
  doi={10.48550/arXiv.2509.14791},
  year={2025}
}

@misc{wada2025tradeoffs,
  title={Trade-offs between Quantum and Classical Resources in the Linear Combination of Unitaries},
  author={Wada, Kaito and Harada, Hiroyuki and Suzuki, Yasunari and Tokunaga, Yuuki and Yamamoto, Naoki and Endo, Suguru},
  eprint={2512.06260},
  archivePrefix={arXiv},
  primaryClass={quant-ph},
  doi={10.48550/arXiv.2512.06260},
  year={2025}
}

@misc{aharonov2025importance,
  title={On the importance of error mitigation for quantum computation},
  author={Aharonov, Dorit and Alberton, Ori and Arad, Itai and Atia, Yosi and Bairey, Eyal and Brakerski, Zvika and Cohen, Itsik and Golan, Omri and Gurwich, Ilya and Kenneth, Oded and others},
  eprint={2503.17243},
  archivePrefix={arXiv},
  primaryClass={quant-ph},
  doi={10.48550/arXiv.2503.17243},
  year={2025}
}

@article{suzuki1990fractal,
  title={Fractal decomposition of exponential operators with applications to many-body theories and Monte Carlo simulations},
  author={Suzuki, Masuo},
  journal={Phys. Lett. A},
  volume={146},
  number={6},
  pages={319--323},
  year={1990},
  publisher={Elsevier},
  doi={10.1016/0375-9601(90)90962-N}
}

@article{suzuki1991general,
  title={General theory of fractal path integrals with applications to many-body theories and statistical physics},
  author={Suzuki, Masuo},
  journal={J. Math. Phys.},
  volume={32},
  number={2},
  pages={400--407},
  year={1991},
  publisher={American Institute of Physics},
  doi={10.1063/1.529425}
}

@article{nielsen2021gate,
  doi={10.22331/q-2021-10-05-557},
  title={Gate set tomography},
  author={Nielsen, Erik and Gamble, John King and Rudinger, Kenneth and Scholten, Travis and Young, Kevin and Blume-Kohout, Robin},
  journal={Quantum},
  volume={5},
  pages={557},
  year={2021},
  publisher={Verein zur F{\"o}rderung des Open Access Publizierens in den Quantenwissenschaften}
}

@misc{greenbaum2015introduction,
  title={Introduction to quantum gate set tomography},
  author={Greenbaum, Daniel},
  eprint={1509.02921},
  archivePrefix={arXiv},
  primaryClass={quant-ph},
  doi={10.48550/arXiv.1509.02921},
  year={2015}
}

@article{childs2012hamiltonian,
  title={Hamiltonian simulation using linear combinations of unitary operations},
  author={Childs, Andrew M. and Wiebe, Nathan},
  journal={Quantum Information and Computation},
  volume={12},
  number={11\&12},
  pages={901--924},
  year={2012},
  publisher={Rinton Press},
  doi={10.26421/QIC12.11-12-1}
}

@article{knill2005quantum,
  title={Quantum computing with realistically noisy devices},
  author={Knill, Emanuel},
  journal={Nature},
  volume={434},
  number={7029},
  pages={39--44},
  year={2005},
  publisher={Nature Publishing Group},
  doi={10.1038/nature03350}
}

@article{wallman2016noise,
  title={Noise tailoring for scalable quantum computation via randomized compiling},
  author={Wallman, Joel J. and Emerson, Joseph},
  journal={Physical Review A},
  volume={94},
  number={5},
  pages={052325},
  year={2016},
  publisher={APS},
  doi={10.1103/PhysRevA.94.052325}
}

@article{takagi2022fundamental,
  title={Fundamental limits of quantum error mitigation},
  author={Takagi, Ryuji and Endo, Suguru and Minagawa, Shintaro and Gu, Mile},
  journal={npj Quantum Information},
  volume={8},
  number={1},
  pages={114},
  year={2022},
  publisher={Nature Publishing Group},
  doi={10.1038/s41534-022-00618-z}
}

@article{tsubouchi2023universal,
  title={Universal Cost Bound of Quantum Error Mitigation Based on Quantum Estimation Theory},
  author={Tsubouchi, Kento and Sagawa, Takahiro and Yoshioka, Nobuyuki},
  journal={Physical Review Letters},
  volume={131},
  number={21},
  pages={210601},
  year={2023},
  publisher={APS},
  doi={10.1103/PhysRevLett.131.210601}
}

@article{takagi2023universal,
  title={Universal Sampling Lower Bounds for Quantum Error Mitigation},
  author={Takagi, Ryuji and Tajima, Hiroyasu and Gu, Mile},
  journal={Physical Review Letters},
  volume={131},
  number={21},
  pages={210602},
  year={2023},
  publisher={APS},
  doi={10.1103/PhysRevLett.131.210602}
}

@article{quek2024exponentially,
  title={Exponentially tighter bounds on limitations of quantum error mitigation},
  author={Quek, Yihui and Stilck Fran{\c c}a, Daniel and Khatri, Sumeet and Meyer, Johannes Jakob and Eisert, Jens},
  journal={Nature Physics},
  volume={20},
  number={10},
  pages={1648--1658},
  year={2024},
  publisher={Nature Publishing Group},
  doi={10.1038/s41567-024-02536-7}
}

@article{flammia2020efficient,
  title={Efficient Estimation of Pauli Channels},
  author={Flammia, Steven T. and Wallman, Joel J.},
  journal={ACM Transactions on Quantum Computing},
  volume={1},
  number={1},
  pages={1--32},
  year={2020},
  publisher={Association for Computing Machinery},
  doi={10.1145/3408039}
}

@article{vandenberg2023probabilistic,
  title={Probabilistic error cancellation with sparse Pauli--Lindblad models on noisy quantum processors},
  author={van den Berg, Ewout and Minev, Zlatko K. and Kandala, Abhinav and Temme, Kristan},
  journal={Nature Physics},
  volume={19},
  number={8},
  pages={1116--1121},
  year={2023},
  publisher={Nature Publishing Group},
  doi={10.1038/s41567-023-02042-2}
}

@article{suzuki2022quantum,
  title={Quantum Error Mitigation as a Universal Error Reduction Technique: Applications from the NISQ to the Fault-Tolerant Quantum Computing Eras},
  author={Suzuki, Yasunari and Endo, Suguru and Fujii, Keisuke and Tokunaga, Yuuki},
  journal={PRX Quantum},
  volume={3},
  number={1},
  pages={010345},
  year={2022},
  publisher={APS},
  doi={10.1103/PRXQuantum.3.010345}
}

@article{piveteau2021error,
  title={Error Mitigation for Universal Gates on Encoded Qubits},
  author={Piveteau, Christophe and Sutter, David and Bravyi, Sergey and Gambetta, Jay M. and Temme, Kristan},
  journal={Physical Review Letters},
  volume={127},
  number={20},
  pages={200505},
  year={2021},
  publisher={APS},
  doi={10.1103/PhysRevLett.127.200505}
}

\onecolumngrid

\clearpage



\onecolumngrid

\clearpage

\begin{center}
    \textbf{\large Appendix}
\end{center}
\appendix

\section{Unified measure of estimation quality: mean-squared error (MSE)}
\label[appendix]{app:mse}

In this appendix, we review the basic properties of the mean-squared error (MSE) and explain why it is a useful unified measure of estimation quality. For an estimator $\hat{\theta}$ of a parameter $\theta$, there are two main sources of error.
First, the estimator may suffer from a \emph{systematic error} (bias), defined by
\begin{equation}
\mathrm{Bias}(\hat{\theta})
=
\mathbb{E}[\hat{\theta}] - \theta.
\end{equation}
Bias measures the extent to which the estimator is centered away from the true value, and therefore quantifies lack of accuracy due to systematic deviation.
Second, the estimator is subject to \emph{statistical error}, measured by its variance,
\begin{equation}
\mathrm{Var}(\hat{\theta})
=
\mathbb{E}\!\left[(\hat{\theta}-\mathbb{E}[\hat{\theta}])^2\right].
\end{equation}
Variance quantifies the dispersion of repeated estimates around their own mean, and therefore reflects statistical precision.
The mean-squared error combines these two contributions into a single measure:
\begin{equation}
\mathrm{MSE}(\hat{\theta})
=
\mathbb{E}\!\left[(\hat{\theta}-\theta)^2\right].
\end{equation}
Expanding around $\mathbb{E}[\hat{\theta}]$ gives
\begin{align}
\mathrm{MSE}(\hat{\theta})
&=
\mathbb{E}\!\left[
\bigl(\hat{\theta}-\mathbb{E}[\hat{\theta}] + \mathbb{E}[\hat{\theta}] - \theta\bigr)^2
\right]
\\
&=
\mathbb{E}\!\left[(\hat{\theta}-\mathbb{E}[\hat{\theta}])^2\right]
+
\bigl(\mathbb{E}[\hat{\theta}] - \theta\bigr)^2
\\
&=
\mathrm{Var}(\hat{\theta}) + \mathrm{Bias}(\hat{\theta})^2.
\end{align}

This decomposition shows that MSE accounts for both random fluctuations and systematic deviation. For this reason, MSE is a natural measure of \emph{overall estimation error}. In particular, variance alone measures precision, whereas MSE evaluates estimator quality when both precision and accuracy are relevant. An estimator with small variance but large bias may still have poor overall performance, and MSE captures this trade-off directly.
When both statistical error and systematic error are present, we therefore adopt the MSE as a unified criterion for comparing estimators.

Because $\mathbb{E}[(\hat{\theta}-\theta)^2] = \mathrm{MSE}(\hat{\theta})$, the Chebyshev inequality gives
\begin{equation}
\Pr\!\bigl(|\hat{\theta}-\theta| \ge a\bigr)
\le
\frac{\mathrm{MSE}(\hat{\theta})}{a^2}.
\end{equation}
Taking the complement,
\begin{equation}
\Pr\!\bigl(|\hat{\theta}-\theta| < a\bigr)
\ge
1 - \frac{\mathrm{MSE}(\hat{\theta})}{a^2}.
\end{equation}
To guarantee at least $1-\alpha$ coverage, set
\begin{equation}
a = \sqrt{\frac{\mathrm{MSE}(\hat{\theta})}{\alpha}},
\end{equation}
which yields
\begin{equation}
\Pr\!\!\left(\theta \in \hat{\theta} \pm \sqrt{\frac{\mathrm{MSE}(\hat{\theta})}{\alpha}}\right)
\ge 1-\alpha.
\end{equation}
Thus an MSE-based confidence bound is valid in this distribution-free, conservative sense:
\begin{equation}
\theta \in \hat{\theta} \pm \sqrt{\frac{\mathrm{MSE}(\hat{\theta})}{\alpha}}.
\end{equation}

\section{Extraction of the Trotter prefactor used in the numerical comparison}
\label[appendix]{app:alpha_fit}

This appendix details how the second-order Trotter prefactor $\alpha_2$ entering
\cref{fig:M_vs_n_eps} is obtained, for both the tight (``Trotter'') and
commutator-bound (``Trotter bound'') curves. Throughout we take $t=n$, normalize
the couplings so that $\|H\|_{\ell_1}=3n$, work with ten random instances, and use
$\gamma'=2\times10^{-7}$.

Because the ideal channel $\mathcal U_t$ and the $N$-step second-order Trotter
channel $\mathcal V_N^{(2)}$ are both unitary, their diamond distance is given in
closed form by the eigenphases of $W=(V_N^{(2)})^{\dagger}U_t$. 
Writing $\{e^{i\theta_j}\}$ for the eigenvalues of $W$ and
$\omega$ for the smallest arc of the unit circle containing all phases $\theta_j$,
\begin{equation}
D(\mathcal U_t,\mathcal V_N^{(2)})=
\begin{cases}
\sin(\omega/2) & \omega<\pi,\\
1 & \omega\ge\pi,
\end{cases}
\end{equation}
which we evaluate exactly for $n\le 12$. For each instance we compute $D$ at
several Trotter numbers $N$ and obtain $\alpha_2(n)$ from a least-squares fit to
$D=\alpha_2/N^{2}$. The linear fit is consistent with the exponent ($R^2\simeq0.99$). The
``Trotter bound'' curve replaces this fitted value by the analytic second-order
commutator prefactor
$\bigl(\|[A,[A,B]]\|/12+\|[B,[B,A]]\|/24\bigr)\,t^{3}$, where $A$ and $B$ collect
the even- and odd-bond terms of $H$.

To reach the larger sizes shown, we extrapolate beyond $n=12$. Because $\alpha^{(2)}_{\mathrm{comm}}=O(L)=O(n)$ for the geometrically local chain and $t=n$, the prefactor scales as $\alpha_2(n)=O(n^{4})$. Dividing out the $t^{3}=n^{3}$ factor leaves a residual $\alpha_2(n)/n^{3}=O(n)$ that is linear in $n$; accordingly we take the instance median of $\alpha_2(n)/n^{3}$, fit it linearly in $n$ over $n\le 12$, and use the fit to extrapolate. The same procedure is applied to $\alpha^{(2)}_{\mathrm{comm}}(n)$.

\section{Proof of \texorpdfstring{\cref{thm:avg_gate_count}}{Theorem 1}}
\label[appendix]{app:proof_avg_gate_count}

\begin{proof}
The mean value of $k$ is
\begin{equation}
\label{eq:avg_gate_exact}
\E_{p}[k]
=
\sum_{k\in{\mathrm{even}}} k\, p(k, \tau),
\end{equation}
where
\begin{equation}
p(k, \tau)=
\begin{cases}
\displaystyle
\frac{\dfrac{\tau^k}{k!}\sqrt{1+\left(\dfrac{\tau}{k+1}\right)^2}}
{\displaystyle\sum_{k'\in\mathrm{even}}\dfrac{\tau^{k'}}{k'!}\sqrt{1+\left(\dfrac{\tau}{k'+1}\right)^2}}
& k\ \mathrm{even},\\[8pt]
0 & k\ \mathrm{odd}.
\end{cases}
\end{equation}
Define the decreasing weight
\begin{equation}
a_k := \sqrt{1+\left(\frac{\tau}{k+1}\right)^2},
\end{equation}
and the even-Poisson reference distribution
\begin{equation}
q(k,\tau) :=
\begin{cases}
\dfrac{\tau^k/k!}{\sum_{k'\in\mathrm{even}}\tau^{k'}/k'!}
=\dfrac{\tau^k/k!}{\cosh \tau}, & k\ \mathrm{even},\\[6pt]
0, & k\ \mathrm{odd}.
\end{cases}
\end{equation}
Then $p$ can be written as
\begin{equation}
p(k,\tau)=\frac{q(k,\tau)\,a_k}{\sum_{j\in\mathrm{even}} q(j,\tau)\,a_j}.
\end{equation}
Because $a_k$ is decreasing in $k$, for any increasing function $f(k)$ we have
\begin{equation}
\E_{p}[f(k)]
=
\frac{\sum_k q(k,\tau)\,f(k)\,a_k}{\sum_k q(k,\tau)\,a_k}
\le
\sum_k q(k,\tau)\,f(k)
=
\E_{q}[f(k)].
\end{equation}
Applying this to $f(k)=k$ and $f(k)=k^2$ gives
\begin{gather}
\E_{p}[k] \le \E_{q}[k],
\qquad
\E_{p}[k^2] \le \E_{q}[k^2].
\end{gather}
It remains to compute the moments under $q$. Using the moment generating function
\begin{equation}
\label{eq:moment}
M_q(t)=\sum_k q(k,\tau)e^{tk}=\frac{\cosh(\tau e^t)}{\cosh\tau},
\end{equation}
we obtain
\begin{equation}
\begin{split}
\E_{q}[k] &= M_q'(0)=\tau\tanh\tau, \\
\E_{q}[k^2] &= M_q''(0)=\tau(\tau+\tanh\tau)=\tau^2+\tau\tanh\tau,
\end{split}
\end{equation}
which proves \cref{eq:mean_k_bound,eq:var_k_bound}.
\end{proof}

\section{Improvement of the simulation accuracy via Trotter extrapolation}
\label[appendix]{Sec:extrapolation}
The impact of algorithmic errors in Trotter-based simulations can be mitigated by extrapolating to $d=\infty$ using results obtained from several Trotter counts $\{d_l \}_{l=1}^m$, where $m$ is the number of points used for extrapolation~\cite{endo2019mitigating,watson2025exponentially}. Recently, Ref.~\cite{watson2025exponentially} shows that the dependence of the Trotter count on the required accuracy $\epsilon$ can be exponentially improved via Richardson extrapolation with polynomial interpolation, i.e., the maximum Trotter count reads:
\begin{equation}
d_{{\rm max}}=\max_l d_l = O((h_{\rm max} \Upsilon_k \Lambda t)^{1+k^{-1}} \log(\epsilon^{-1}) ), 
\end{equation}
with $m= O(\log(\epsilon^{-1}))$.
Here, $h_{\rm max}$ denotes the maximum absolute value of the time-step coefficients used in the product formula, and $\Lambda$ encapsulates the commutator scaling of the system, which is strictly upper-bounded as $\Lambda \le 4\sum_{\ell} \|H_{\ell}\|$. For the PEC to work efficiently, i.e., the sampling overhead is $O(1)$, we need $\gamma' L\, d_{\rm max}=O(1)$. At this point, we have $\epsilon_0= \exp[- (\gamma' L )^{-1} (h_{\rm max} \Upsilon_k \Lambda t)^{-(1+k^{-1})}] $, which is superpolynomially small in $\gamma' L$ compared with the critical error $\epsilon_c$ in the straightforward Trotter simulation with PEC. Even for $\epsilon < \epsilon_0$, the PEC sampling overhead only scales polynomially with $\epsilon^{-1}$, which indicates there does not exist the critical error $\epsilon_c$ for this method.

\section{Proof of \texorpdfstring{\cref{thm:avg_pec_overhead}}{Theorem 2}}
\label[appendix]{app:proof_avg_pec_overhead}

\begin{proof}
Under PEC the estimator is unbiased for the ideal expectation value, hence the MSE reduces to the variance.
Consider one shot of RLCU with PEC\@.
Let $\hat{\nu}$ be the measured outcome of the bounded observable in that shot, so that $|\hat{\nu}|\le 1$.
The shot is rescaled by the RLCU sampling overhead and the PEC overhead, so the single-shot random variable can be written as
\begin{equation}
    Z \;=\; \Gamma_{\mathrm{RLCU}}\;\Gamma_{\mathrm{PEC}}(\omega)\;\hat{\nu},
\end{equation}
where $\omega$ denotes the randomness in the sampled LCU terms and in PEC.
With $M$ independent shots, $\hat{\mu}=\frac{1}{M}\sum_{j=1}^M Z_j$, and therefore
\begin{equation}
    \epsilon^2 
    \;=\; \frac{1}{M}\mathrm{Var}(Z)
    \;\le\; \frac{1}{M}\E[Z^2]
    \;\le\; \frac{\Gamma_{\mathrm{RLCU}}^{2}}{M}\E\!\left[\Gamma_{\mathrm{PEC}}(\omega)^{2}\right],
\end{equation}
using $\mathrm{Var}(Z)\le \E[Z^2]$ and $|\hat{\nu}|\le 1$.
For RLCU we use the standard bound
\begin{equation}
    \Gamma_{\mathrm{RLCU}}^{2}\;\le\;\exp\!\left(\frac{4\tilde t^{2}}{r}\right).
\end{equation}

It remains to bound $\E[\Gamma_{\mathrm{PEC}}(\omega)^2]$.
For a given realization, segment $j$ applies both control unitaries $U_{\mu_j}$ and $U_{\nu_j}$ of $G_{\mu\nu}$, whose Taylor orders $k_{\mu_j}$ and $k_{\nu_j}$ are even integers sampled from $p(k,\tau)$ with $\tau=\tilde t/r$; the two orders within a segment need not be independent, but the pairs are i.i.d.\ across segments.
In our gate model each control unitary consists of a single non-Clifford rotation together with its Clifford Paulis, so segment $j$ contributes $(1+k_{\mu_j})+(1+k_{\nu_j})$ elementary operations and the total gate count is
\begin{equation}
    N_g(\omega)\;=\;\sum_{j=1}^r\bigl[(1+k_{\mu_j})+(1+k_{\nu_j})\bigr]
    \;=\;2r+\sum_{j=1}^r\bigl(k_{\mu_j}+k_{\nu_j}\bigr).
\end{equation}
The first term counts the $2r$ non-Clifford gates and the second the Clifford gates.
Assigning the non-Clifford contribution to $\gamma'$ and the Clifford contributions to $\gamma_c'$, the PEC sampling overhead satisfies
\begin{equation}
    \Gamma_{\mathrm{PEC}}(\omega)^2
    \le
    \exp\!\left(4\gamma' r + 2\gamma_c' \sum_{j=1}^r (k_{\mu_j}+k_{\nu_j})\right).
\end{equation}
Taking the expectation and using that the pairs are i.i.d.\ across segments,
\begin{equation}
    \E\!\left[\Gamma_{\mathrm{PEC}}(\omega)^2\right]
    \le
    \exp(4\gamma' r)\;\prod_{j=1}^r \E\!\left[e^{2\gamma_c'(k_{\mu_j}+k_{\nu_j})}\right].
\end{equation}
Within a segment $k_{\mu_j}$ and $k_{\nu_j}$ are identically distributed, so by the Cauchy--Schwarz inequality
\begin{equation}
    \E\!\left[e^{2\gamma_c'(k_{\mu_j}+k_{\nu_j})}\right]
    \le
    \sqrt{\E\!\left[e^{4\gamma_c' k_{\mu_j}}\right]\,\E\!\left[e^{4\gamma_c' k_{\nu_j}}\right]}
    =
    \E_{p}\!\left[e^{4\gamma_c' k}\right].
\end{equation}
Therefore,
\begin{align}
    \E\!\left[\Gamma_{\mathrm{PEC}}(\omega)^2\right]
    &\le
    \exp(4\gamma' r)\;
    \Bigl(\E_{p}\!\left[e^{4\gamma_c' k}\right]\Bigr)^{r}.
\end{align}

To bound $\E_{p}[e^{4\gamma_c' k}]$, we use the same argument as in \cref{thm:avg_gate_count}.
Write $p(k,\tau)\propto q(k,\tau)a_k$ where $q(k,\tau)$ is the even-Poisson reference distribution and
$a_k=\sqrt{1+(\tau/(k+1))^2}$ is decreasing in $k$.
Because $e^{4\gamma_c' k}$ is increasing in $k$, this implies
\begin{equation}
    \E_{p}\!\left[e^{4\gamma_c' k}\right]\;\le\;\E_{q}\!\left[e^{4\gamma_c' k}\right].
\end{equation}

From~\cref{eq:moment}, we have
\begin{align}
\E_{q}\!\left[e^{4\gamma_c' k}\right] = 
    \frac{\cosh(\tau e^{4\gamma_c'})}{\cosh(\tau)}
    \le
    e^{\tau(e^{4\gamma_c'}-1)}.
\end{align}

Plugging this into the previous bound and using $r\tau=\tilde t$ gives
\begin{equation}
    \E\!\left[\Gamma_{\mathrm{PEC}}(\omega)^2\right]
    \;\le\;
    \exp\!\left(
        4\gamma' r
        +\tilde t\bigl(e^{4\gamma_c'}-1\bigr)
    \right).
\end{equation}

Combining with the RLCU overhead bound yields
\begin{equation}
    \epsilon^2
    \;\le\;
    \frac{1}{M}
    \exp\!\left(
        \frac{4\tilde t^{2}}{r}
        +4\gamma' r
        +\tilde t\bigl(e^{4\gamma_c'}-1\bigr)
    \right),
\end{equation}
which proves \cref{eq:phi_r_def}.
\end{proof}

\section{Derivation of the segmented SNI scaling}
\label[appendix]{app:sni_segment_scaling}

\subsection{Setup and notation}

We partition the $d$-layer Trotter circuit into $s$ segments of equal depth $d/s$.
Each segment contains $dL/s$ elementary gates.
Let $q \le 1-{(1-\gamma)}^{L} \simeq \gamma L$ denote the error probability per layer,
where $\gamma$ is the per-gate error rate.
Under the stochastic Pauli noise assumption obtained via twirling,
the per-segment error probability satisfies the bound
\begin{equation}
    \label{eq:q_st_ub_app}
    q_{\rm ST} \;\le\; \frac{qd}{s}.
\end{equation}
The SNI method is applicable when $q_{\rm ST} < 1/2$,
which requires $s > 2qd$.

Applying SNI independently to each segment,
the total sampling overhead is $(1-2q_{\rm ST})^{-2s}$,
and the number of circuit runs required for the
error-mitigated computation satisfies~\cref{eq:M_sni_segment}.
The noise-characterization cost for estimating $q_{\rm ST}$
at each of the $s$ segments is
\begin{equation}
    \label{eq:Mqst_app}
    M_{q_{\rm ST}}
    = O\!\left(\epsilon^{-2}\,s^{2}\,(1-2q_{\rm ST})^{-4}\right),
\end{equation}
where the factor $s^{2}$ arises from requiring per-segment accuracy $\epsilon/s$
to ensure overall accuracy $\epsilon$ via the union bound.

\subsection{Bounding the total overhead}

We derive an upper bound on the per-segment overhead $(1-2q_{\rm ST})^{-s}$.
Taking the natural logarithm and applying the inequality
$-\ln(1-x) \le \frac{x}{1-x}$ for $0 < x < 1$ with $x = 2q_{\rm ST}$,
\begin{equation}
    \label{eq:log_ineq_app}
    \ln(1-2q_{\rm ST})^{-s}
    = -s\ln(1-2q_{\rm ST})
    \le \frac{2s\,q_{\rm ST}}{1-2q_{\rm ST}}.
\end{equation}
The function $g(y) = \frac{2y}{1-2y}$ is monotone increasing for $0 < y < 1/2$.
Because $q_{\rm ST} \le qd/s$ from~\cref{eq:q_st_ub_app}, substituting the upper bound gives
\begin{equation}
    \frac{2s\,q_{\rm ST}}{1-2q_{\rm ST}}
    \le \frac{2s \cdot (qd/s)}{1-2(qd/s)}
    = \frac{2qd}{1-2qd/s}.
\end{equation}
Exponentiating yields the upper bound
\begin{equation}
    \label{eq:sni_ub}
    (1-2q_{\rm ST})^{-s}
    \le \exp\!\left(\frac{2qd}{1-2qd/s}\right).
\end{equation}

\subsection{Lower bound on the overhead}

We also establish the matching lower bound quoted in the main text. Writing $n:=d/s$ and using the exact per-segment error probability $q_{\rm ST}=1-(1-q)^{n}$ from~\cref{eq:qst_seg_bound},
\begin{equation}
    1-2q_{\rm ST} = 2(1-q)^{n}-1 \le 2e^{-qn}-1 < e^{-2qn},
\end{equation}
where the first inequality uses $(1-q)^{n}\le e^{-qn}$ and the second follows from $2u-1<u^{2}$ for $u=e^{-qn}\in(0,1)$ (equivalently $(u-1)^{2}>0$). In the regime $q_{\rm ST}<1/2$ where SNI applies, $0<1-2q_{\rm ST}<e^{-2qn}$, so raising to the power $-2s$ reverses the inequality and gives
\begin{equation}
    (1-2q_{\rm ST})^{-2s} > e^{4qns} = e^{4qd}.
\end{equation}
Combined with the square of~\cref{eq:sni_ub}, this yields the bound $e^{4qd}<(1-2q_{\rm ST})^{-2s}<\exp\!\left(4qd/(1-2qd/s)\right)$ used in the main text.

\section{Bounding the bias of the expectation values under noise and PEC}
\label[appendix]{app:bias_bound}

Let $\{\mathcal{U}_k\}_{k=1}^{N_G}$ denote the ideal unitary channels implementing the target
circuit, and let $\{\mathcal{E}_k\}_{k=1}^{N_G}$ be noise channels. Note that we later consider the cases where $\mathcal{E}_k$ is a simple error described by the CPTP channel and the $\mathcal{E}_k$ corresponds to the residual error due to the incomplete noise characterization for PEC. Define
\begin{equation}
\mathcal{C}_{\mathrm{noisy}} := \prod_{k=1}^{N_G} \mathcal{E}_k\,\mathcal{U}_k,
\qquad
\mathcal{C}_{\mathrm{ideal}} := \prod_{k=1}^{N_G} \mathcal{U}_k,
\end{equation}
where $\prod_{k=1}^{N_G}\mathcal{A}_k := \mathcal{A}_{N_G}\circ\cdots\circ\mathcal{A}_1$.
We measure deviations with the channel distance
\begin{equation}
D(\mathcal{E},\mathcal{F}) := \tfrac{1}{2}\|\mathcal{E}-\mathcal{F}\|
\end{equation}
induced by a submultiplicative norm $\|\cdot\|$ satisfying
$\|\Lambda\circ\Gamma\|\le\|\Lambda\|\,\|\Gamma\|$
and $\|\mathcal{U}\|=1$ for any unitary channel $\mathcal{U}$
(e.g., the diamond norm or an induced trace distance norm).

Setting $\Phi_k := \mathcal{E}_k\circ\mathcal{U}_k$ and $\Psi_k := \mathcal{U}_k$,
a telescoping identity gives
\begin{equation}
\prod_{k=1}^{N_G}\Phi_k - \prod_{k=1}^{N_G}\Psi_k
= \sum_{j=1}^{N_G}
\Bigl(\prod_{k=j+1}^{N_G}\Phi_k\Bigr)\circ(\Phi_j-\Psi_j)\circ
\Bigl(\prod_{k=1}^{j-1}\Psi_k\Bigr).
\end{equation}
Applying the triangle inequality and submultiplicativity yields the general bound
\begin{equation}
D\!\left(\mathcal{C}_{\mathrm{noisy}},\mathcal{C}_{\mathrm{ideal}}\right)
\le \sum_{j=1}^{N_G}
\left(\prod_{k=j+1}^{N_G}\|\mathcal{E}_k\|\right) D(\mathcal{E}_j,\mathcal{I}).
\label{eq:general_telescoping_bound}
\end{equation}

\paragraph*{Ordinary CPTP noise.}
With no error mitigation, set $\mathcal{E}_k = \mathcal{N}_k$ where each $\mathcal{N}_k$ is a
CPTP gate noise channel. Introducing the per-gate noise strength and norm,
\begin{equation}
\gamma := \max_k 2D(\mathcal{N}_k,\mathcal{I})=\max_k\|\mathcal{N}_k-\mathcal{I}\|_\diamond,
\qquad
M_{\mathrm{noise}} := \max_k\|\mathcal{N}_k\|,
\end{equation}
\cref{eq:general_telescoping_bound} gives
$D(\mathcal{C}_{\mathrm{noisy}},\mathcal{C}_{\mathrm{ideal}})
\le \tfrac{\gamma}{2}\sum_{m=0}^{N_G-1}M_{\mathrm{noise}}^{m}$.
For the diamond norm or an induced $1\to1$ norm, CPTP maps satisfy $M_{\mathrm{noise}}=1$,
yielding the linear accumulation bound
\begin{equation}
D(\mathcal{C}_{\mathrm{noisy}},\mathcal{C}_{\mathrm{ideal}}) \le \frac{N_G\,\gamma}{2}.
\label{eq:ordinary_noise_bound}
\end{equation}

\paragraph*{Residual noise after PEC.}
When probabilistic error cancellation is applied using an estimated noise model $\{\mathcal{N}'_k\}$,
the effective channel on gate $k$ becomes
$\mathcal{N}_k^{\prime-1}\circ\mathcal{N}_k\circ\mathcal{U}_k$.
The residual perturbation channel is therefore
\begin{equation}
\Delta\mathcal{N}_k := \mathcal{N}_k^{\prime-1}\circ\mathcal{N}_k.
\end{equation}
Setting $\mathcal{E}_k = \Delta\mathcal{N}_k$ in \cref{eq:general_telescoping_bound}
and introducing
\begin{equation}
\Delta\gamma := \max_{1\le k\le N_G} D(\Delta\mathcal{N}_k,\mathcal{I}),
\end{equation}
\cref{eq:general_telescoping_bound} implies the uniform bound
\begin{equation}
D\!\left(\mathcal{C}_{\mathrm{noisy}},\mathcal{C}_{\mathrm{ideal}}\right)
\le \Delta\gamma \sum_{j=1}^{N_G}\prod_{k=j+1}^{N_G}\|\Delta\mathcal{N}_k\|.
\label{eq:Dgammag_uniform_bound}
\end{equation}
Letting $M_{\Delta} := \max_k\|\Delta\mathcal{N}_k\|$, \cref{eq:Dgammag_uniform_bound} gives
\begin{equation}
D\!\left(\mathcal{C}_{\mathrm{noisy}},\mathcal{C}_{\mathrm{ideal}}\right)
\le \Delta\gamma\sum_{m=0}^{N_G-1} M_{\Delta}^{m}
= \Delta\gamma\,\frac{M_{\Delta}^{N_G}-1}{M_{\Delta}-1}, \qquad (M_{\Delta}\neq 1).
\label{eq:geom_bound}
\end{equation}
In the small-residual regime relevant to PEC, $\|\Delta\mathcal{N}_k-\mathcal{I}\|\ll 1$ implies
\begin{equation}
\|\Delta\mathcal{N}_k\|
\le \|\mathcal{I}\| + \|\Delta\mathcal{N}_k-\mathcal{I}\|
= 1 + 2D(\Delta\mathcal{N}_k,\mathcal{I})
\le 1 + 2\Delta\gamma.
\end{equation}
Substituting into \cref{eq:Dgammag_uniform_bound} gives
\begin{equation}
D\!\left(\mathcal{C}_{\mathrm{noisy}},\mathcal{C}_{\mathrm{ideal}}\right)
\le \Delta\gamma\sum_{m=0}^{N_G-1}(1+2\Delta\gamma)^m
= \frac{(1+2\Delta\gamma)^{N_G}-1}{2}.
\label{eq:small_Dgammag_bound}
\end{equation}
For $\Delta\gamma\ll 1$, expanding \cref{eq:small_Dgammag_bound} yields
\begin{equation}
D\!\left(\mathcal{C}_{\mathrm{noisy}},\mathcal{C}_{\mathrm{ideal}}\right)
= N_G\Delta\gamma + O(N_G^2\Delta\gamma^2),
\end{equation}
so the circuit-level deviation scales approximately linearly with the number of gates.
In the regime of interest where the bias is small, the leading term dominates.

\paragraph*{Bounding $\Delta\gamma$ via GST.}
Let $\Delta\gamma'_g := \max_k D(\mathcal{N}_k,\mathcal{N}'_k)$ denote the per-gate
characterization error, and suppose $D(\mathcal{N}'_k,\mathcal{I})<\mu$ for all $k$.
Because $\mathcal{N}_k^{\prime-1}\circ\mathcal{N}'_k = \mathcal{I}$, for each $k$,
\begin{align}
D(\Delta\mathcal{N}_k,\mathcal{I})
&= D(\mathcal{N}_k^{\prime-1}\circ\mathcal{N}_k,\,\mathcal{N}_k^{\prime-1}\circ\mathcal{N}'_k)
\notag \\
&= \tfrac{1}{2}\|\mathcal{N}_k^{\prime-1}\circ(\mathcal{N}_k-\mathcal{N}'_k)\|
\le \|\mathcal{N}_k^{\prime-1}\|\,D(\mathcal{N}_k,\mathcal{N}'_k).
\end{align}
Because $\|\mathcal{N}'_k-\mathcal{I}\| = 2D(\mathcal{N}'_k,\mathcal{I}) < 2\mu$,
a Neumann series argument gives, for $\mu<1/2$,
\begin{equation}
\|\mathcal{N}_k^{\prime-1}\|
= \|(\mathcal{I}+(\mathcal{N}'_k-\mathcal{I}))^{-1}\|
\le \frac{1}{1-\|\mathcal{N}'_k-\mathcal{I}\|}
\le \frac{1}{1-2\mu}.
\end{equation}
Therefore,
\begin{equation}
\Delta\gamma \le \frac{\Delta\gamma'_g}{1-2\mu}, \qquad (\mu<\tfrac{1}{2}).
\end{equation}
Because GST with $M_g$ samples achieves $\Delta\gamma'_g = O(1/\sqrt{M_g})$,
we obtain $\Delta\gamma = O(1/\sqrt{M_g})$.
Combined with the leading-order bound, the bias in the expectation value satisfies
\begin{equation}
\label{eq:bias_bound_general_app}
\big|\mathrm{Bias}(\hat{\mu})\big|
= 
\bigl|\Tr[O\,\mathcal{C}_{\mathrm{noisy}}(\rho)]
- \Tr[O\,\mathcal{C}_{\mathrm{ideal}}(\rho)]\bigr|
= O\!\left(\frac{N_G}{\sqrt{M_g}}\right)
\end{equation}
for any observable $O$ with $\|O\|_\infty\le 1$.

For the RLCU algorithm, $\mathcal{C}_{\mathrm{noisy}}$ and $\mathcal{C}_{\mathrm{ideal}}$
are randomly generated circuits, and the same argument applies.
The estimator involves rescaling by the sampling overhead $\Gamma_{\mathrm{RLCU}}$,
so the bias bound becomes $\Gamma_{\mathrm{RLCU}} N_G\gamma$ without error mitigation
and $\Gamma_{\mathrm{RLCU}} N_G\Delta\gamma$ after PEC.

\end{document}